\documentclass[letterpaper, journal, twoside]{IEEEtran}

 \IEEEpubidadjcol 
\IEEEoverridecommandlockouts

\usepackage{graphics} 
\usepackage{epsfig} 
\usepackage{amsmath} 
\usepackage{amssymb}  

\usepackage{algpseudocode}
\usepackage{algorithm}
\usepackage{epstopdf}

\usepackage{verbatim}

\usepackage{amsthm}

\usepackage{color}
\usepackage{cancel}

\usepackage{enumitem}

\usepackage{tikz}               
\usepgflibrary{arrows}

\usepackage{hyperref}

\usepackage[normalem]{ulem}

\newtheorem{assumption}{\bf Assumption}

\newtheorem{theorem}{\bf Theorem}
\newtheorem{proposition}{\bf Proposition}
\newtheorem{lemma}{\bf Lemma}
\newtheorem{corollary}{\bf Corollary}
\newtheorem{remark}{\bf Remark}

\title{Constrained nonlinear output regulation using model predictive control - extended version}
\author{Johannes K\"ohler$^1$, Matthias A. M\"uller$^2$, Frank Allg\"ower$^1$
\thanks{$^1$Johannes K\"ohler and Frank Allg\"ower are with the
Institute for Systems Theory and Automatic Control, University of Stuttgart,
70550 Stuttgart, Germany. (email:$\{$johannes.koehler, frank.allgower$\}$@ist.uni-stuttgart.de).}
\thanks{$^2$Matthias A. M\"uller is with the Institute of Automatic Control, Leibniz University Hannover, 30167 Hannover, Germany.
(email:mueller@irt.uni-hannover.de).}
\thanks{Johannes K\"ohler would like to thank the German Research Foundation (DFG) for financial support of the project within the International Research Training Group “Soft Tissue Robotics” (GRK 2198/1 - 277536708).} 
}

\begin{document}
\IEEEoverridecommandlockouts
\IEEEpubid{\begin{minipage}{\textwidth}\ \\[20pt] \\ \\
         \copyright 2021 IEEE.  Personal use of this material is  permitted.  Permission from IEEE must be obtained for all other uses, in  any current or future media, including reprinting/republishing this material for advertising or promotional purposes, creating new  collective works, for resale or redistribution to servers or lists, or  reuse of any copyrighted component of this work in other works.
     \end{minipage}}
\maketitle
\begin{abstract}
We present a model predictive control (MPC) framework to solve the constrained nonlinear output regulation problem. 
The main feature of the proposed framework is that the application does \textit{not} require the solution to  classical regulator  (Francis-Byrnes-Isidori) equations or any other offline design procedure. 
In particular, the proposed formulation simply minimizes the predicted output error, possibly with some input regularization.
Instead of using terminal cost/sets or a positive definite stage cost as is standard in MPC theory, we build on the theoretical results by Grimm et al.~\cite{grimm2005model} using a \textit{detectability} notion. 
The proposed formulation is applicable if the constrained nonlinear regulation problem is (strictly) feasible, the plant is incrementally stabilizable and incrementally input-output to state stable (i-IOSS/detectable). 
We show that for minimum phase systems such a design ensures exponential stability of the regulator manifold. 
We also provide a design procedure in case of unstable zero dynamics using an incremental input regularization and a nonresonance condition. 
Inherent robustness properties for the noisy error/output-feedback case are established under simplifying assumptions (e.g. no state constraints).
The theoretical results are illustrated with an example involving offset free tracking with noisy error feedback. 
The paper also contains novel results for MPC without terminal constraints with positive semidefinite input/output stage costs that are of independent interest.

This paper is an extended version of the accepted paper~\cite{koehler2021Reg}, and contains the following additional results: Exponential bounds on the suboptimality index $\alpha_N$ using an observability condition (App.~\ref{app:obs}) and an extension of the derived theory to the noisy error feedback case (App.~\ref{app:error_feedback}). 
\end{abstract}
 

\section{Introduction}
\subsubsection*{Motivation}
Output regulation is one of the fundamental problems in control theory, combining dynamic trajectory tracking, disturbance rejection and output feedback in a common framework~\cite{isidori1990output,castillo1993nonlinear,pavlov2006uniform,byrnes2012output}, compare also the (robust) servomechanism problem~\cite{davison1976robust}.
The classical solution is to solve the regulator/Francis-Byrnes-Isidori (FBI) equations~\cite{isidori1990output,castillo1993nonlinear}. 
This reduces the problem to the stabilization of a dynamic reachable state and input trajectory, which can, e.g., be studied using the notion of convergent dynamics~\cite{pavlov2006uniform}. 
Alternatively, the plant can be augmented using the \textit{internal model principle}~\cite{francis1976internal}. 
This approach directly lends itself to the error feedback case and can also be applied to nonlinear systems using  an \textit{immersion property}  and an analytical description of the zero-dynamics (cf.~\cite{isidori2013zero}), compare~\cite{byrnes2003limit,marconi2004non,priscoli2009dissipativity}.
Overall, the classical solutions to the nonlinear output regulation problem require a non-trivial offline design procedure, in particular solving a partial differential equation~\cite{isidori1990output}, which is a bottleneck for practical implementation.
In this paper, we present a  model predictive control (MPC)~\cite{rawlings2017model} approach that solves the output regulation problem and does \textit{not} require any offline design such as, e.g., solving the regulator equations.

\subsubsection*{Related work}
Stabilization of a given steady-state using MPC is a largely solved problem, with approaches based on terminal ingredients (terminal set/cost)~\cite{rawlings2017model} or sufficiently long prediction horizons~\cite{grune2012nmpc,reble2012unconstrained,boccia2014stability}. 
Similar methods can be applied to dynamic problems in case a dynamically feasible state and input trajectory is given, compare~\cite{koehler2020quasi,faulwasser2015nonlinear} and~\cite{kohler2018nonlinear}. 
In the output regulation problem, typically only an output reference is known and hence such approaches are not directly applicable. 
The special case of constant exogenous signals is often studied in MPC under the rubric of offset-free tracking or setpoint tracking. 
Existing solutions compute the optimal steady-state offline/online~\cite{magni2001output,limon2018nonlinear}, use velocity formulations~\cite{betti2013robust,magni2005solution} or deploy disturbance observers~\cite{muske2002disturbance,morari2012nonlinear}, to reduce the problem to the stabilization of a given steady-state.
In case of exogenous signals with a known period length $T$, the output regulation problem can be solved by computing the optimal $T$-periodic trajectory offline/online~\cite{falugi2013tracking}/\cite{kohler2020nonlinear}. 
In~\cite{aguilar2014model}, output regulation is studied using a local (polynomial) approximation to the regulator equations and the dynamic programming equations, but no closed-loop guarantees are obtained. 
 %
In summary, the existing approaches reduce the output regulation problem to the stabilization of a given state and input trajectory by explicitly computing the solution to the regulator equations online or offline. 
In the proposed framework, we can drop this requirement since we use an analysis based on the \textit{detectability} notion from~\cite{grimm2005model}, which requires neither a positive definite stage cost nor terminal ingredients. 
Preliminary results in this direction, using a restrictive \textit{nonsingular input cost} condition, can be found in the conference paper~\cite{Kohler2020OutputRegulation}, which partially overlap with the results in Section~\ref{sec:reg}.

\subsubsection*{Contribution}
We present an MPC framework to solve the constrained nonlinear output regulation problem \textit{without} any offline design procedure, such as for example solving the regulator equations. 
We consider the following structural conditions:
a) the regulator equations admit a strictly feasible solution, b) the plant  is incrementally stabilizable, c) the plant is incrementally input-output to state stable (i-IOSS, detectable).  
As a preliminary result, we consider an MPC scheme with an input/output stage cost that minimizes the predicted output $y$ and uses an input regularization w.r.t. the (typically unknown)  feedforward input $\pi_u(w)$, using tools from~\cite{grimm2005model} (Sec.~\ref{sec:reg}).  
Then we present two MPC schemes that ensure exponential stability of the regulator manifold and constrained satisfaction without solving the regulator equations.
First, we show that simply minimizing the predicted output $y$ over a sufficiently long prediction horizon $N$ solves the constrained nonlinear output regulation problem, if the system is minimum phase, i.e., has stable zero dynamics (Sec.~\ref{sec:minPhase}). 
Second, in case of $T$-periodic exogenous signals, we obtain the same theoretical properties for unstable zero dynamics by including an incremental input regularization in the MPC formulation, assuming a nonresonance condition holds (Sec.~\ref{sec:input}). %
We also establish inherent robustness properties of the proposed MPC framework in case of noisy error feedback, given some simplifying assumptions (e.g., no state constraints). 
We demonstrate the applicability and simplicity of the proposed MPC approach using a numerical example.
Overall, the rigorous theoretical guarantees in combination with the fact that no complex design procedure is required for the implementation, makes the proposed MPC framework suitable for practical application.
As a separate contribution, we extend the analysis from~\cite{grimm2005model} for MPC with positive semidefinite stage costs using an observability condition to obtain exponential bounds on the suboptimality index $\alpha$ similar to~\cite[Variant~2]{grune2012nmpc}.

\subsubsection*{Outline} 
Section~\ref{sec:reg} presents the basic output regulation MPC, including a theoretical analysis. 
Section~\ref{sec:minPhase} shows that in case of minimum phase systems,  the stability properties from Section~\ref{sec:reg} remain valid without any input regularization. 
Section~\ref{sec:input} presents an incremental input regularization for periodic exogenous signals, assuming a nonresonance condition holds. 
Section~\ref{sec:linear} discusses the special case of linear systems. 
Section~\ref{sec:num}  demonstrates the applicability of the proposed approach with a numerical example.
Section~\ref{sec:sum} concludes the paper.
In Appendix~\ref{app:obs}, we  extend the analysis from~\cite{grimm2005model} using an observability condition to obtain exponential bounds on the suboptimality index $\alpha_N$ similar to~\cite[Variant~2]{grune2012nmpc}.
In Appendix~\ref{app:error_feedback}, we extend the analysis to  noisy error feedback. 

\subsubsection*{Notation}
For symmetric matrices $A=A^\top\in\mathbb{R}^{n\times n }$  the maximal and minimal eigenvalue are denoted by $\lambda_{\max}(A)$, $\lambda_{\min}(A)$, respectively. 
The quadratic norm with respect to a positive definite matrix $Q=Q^\top$ is denoted by $\|x\|_Q^2:=x^\top Q x$. 
The positive real numbers are denoted by $\mathbb{R}_{\geq 0}=\{r\in\mathbb{R}|~r\geq 0\}$. 
For vectors $x,y$, we abreviate $[x^\top,y^\top]^\top=(x,y)$.

\section{Constrained Output Regulation \\ Amidst Classical results and MPC}
\label{sec:reg}
This section presents the basic output regulation MPC scheme using an input regularization that requires knowledge of the feedforward input $\pi_u(w)$. 
This scheme and the corresponding analysis is the basis for Sections~\ref{sec:minPhase} and \ref{sec:input}, which provide MPC schemes that do \textit{not} require the solution to the regulator equations.
In Section~\ref{sec:reg_1}, we present the output regulation problem, including classical results. 
The proposed MPC scheme is presented in Section~\ref{sec:reg_2} and the theoretical analysis is contained in Section~\ref{sec:reg_3}. 
%
%
\subsection{Output regulation - setup and classical results}
\label{sec:reg_1}
We consider the following nonlinear discrete-time system
\begin{subequations}
\label{eq:sys}
\begin{align}
\label{eq:sys_dyn}
x^p_{t+1}&=f^p(x^p_t,u_t,w_t),\\
\label{eq:sys_exo}
w_{t+1}&=s(w_t),\\
\label{eq:sys_output}
y_t&=h(x^p_t,u_t,w_t), 
\end{align}
\end{subequations}
with $f^p$, $s$, $h$ continuous.
The plant dynamics are described by equation~\eqref{eq:sys_dyn} with the plant state $x^p\in\mathbb{X}^p=\mathbb{R}^{n_p}$ and the control input $u\in\mathbb{U}\subseteq\mathbb{R}^m$. 
The exogenous signal $w\in\mathbb{R}^q$ is generated by the exosystem~\eqref{eq:sys_exo} and represents both disturbances affecting the plant~\eqref{eq:sys_dyn} and desired reference values~\eqref{eq:sys_output}.
Equation~\eqref{eq:sys_output} describes a reference tracking error $y\in\mathbb{R}^p$, which is, e.g., the difference between the plant output and some output reference.
The exosystem state is assumed to be contained in some compact positive invariant set $\mathbb{W}$, i.e., $s:\mathbb{W}\rightarrow\mathbb{W}$. 
The control goal is to achieve output nulling ($\lim_{t\rightarrow\infty}\|y_t\|=0$), while the plant state and control input are supposed to satisfy general nonlinear constraints of the form $(x^p_t,u_t)\in\mathcal{Z}^p\subseteq\mathbb{X}^p\times\mathbb{U}\subseteq\mathbb{R}^{n_p+m}$. 
The classical solution to the output regulation problem is to find functions $\pi_x:\mathbb{W}\rightarrow\mathbb{X}^p$, $\pi_u:\mathbb{W}\rightarrow\mathbb{U}$, which satisfy
\begin{subequations}
\label{eq:reg}
\begin{align}
\label{eq:reg1}
\pi_x(s(w))=&f^p(\pi_x(w),\pi_u(w),w),\\
\label{eq:reg2}
0=&h(\pi_x(w),\pi_u(w),w),
\end{align} 
\end{subequations}
for any $w\in\mathbb{W}$. 
Equations~\eqref{eq:reg} are called the discrete-time regulator equations or Francis-Byrnes-Isidori (FBI) equations.
In~\cite[Thm.~2]{castillo1993nonlinear} it was shown that the regulator equations are locally solvable,  if the zero dynamics of the plant is hyperbolic and the exosystem is neutrally stable.\footnote{%
The zero dynamics is hyperbolic, if the eigenvalues of the Jacobian linearization do not lie on the unit circle.
If the exosystem is neutrally stable, then the eigenvalues of its Jacobian linearization lie on the unit circle.}
 Given a solution to the regulator equations~\eqref{eq:reg}, output regulation can be reduced to the problem of stabilizing the error $e:=x^p-\pi_x(w)$.
\begin{assumption}
\label{ass:regulator} (Regulator equations)
The regulator equations~\eqref{eq:reg} admit a continuous solution $\pi_x$, $\pi_u$ and $(\pi_x(w),\pi_u(w))\in\text{int}(\mathcal{Z}^p)$ for all $w\in\mathbb{W}$. 
\end{assumption}
Uniqueness of $\pi_x$ for a given $\pi_u$ follows from a detectability condition posed latter. 
Uniqueness of both, $\pi_x$ and $\pi_u$, will be ensured in Sections~\ref{sec:minPhase} and~\ref{sec:input} based on an additional minimum phase or nonresonance condition.
\begin{assumption}
\label{ass:increm_stab} (Local incremental exponential stabilizability)  
There exist a continuous incremental Lyapunov function $V_{s}:\mathbb{X}^p\times\mathbb{X}^p\times\mathbb{W}\rightarrow\mathbb{R}_{\geq 0}$,
 a Lipschitz continous feedback $\kappa:\mathbb{X}^p\times\mathbb{X}^p\times\mathbb{W}\times\mathbb{U}\rightarrow\mathbb{U}$ with $V_s(z^p,z^p,w)=0$, $\kappa(z^p,z^p,w,v)=v$, and constants $c_{s,l}$, $c_{s,u}$, $\delta_{loc}>0$, $\rho_s\in(0,1)$, such that for any $(z^p,v,w)\in\mathcal{Z}^p\times\mathbb{W}$ and all $x^p\in\mathbb{X}^p$ satisfying $V_{s}(x^p,z^p,w)\leq \delta_{loc}$, the following inequalities hold:
\begin{subequations}
\label{eq:increm}
\begin{align}
\label{eq:increm_a}
&V_{s}(f^p(x^p,\kappa(x^p,z^p,w,v),w),f^p(z^p,v,w),s(w))\nonumber\\
\leq &\rho_s V_{s}(x^p,z^p,w),\\
\label{eq:increm_b}
&c_{s,l}\|x^p-z^p\|^2\leq  V_{s}(x^p,z^p)\leq c_{s,u}\|x^p-z^p\|^2.
\end{align}
\end{subequations}
\end{assumption} 
This assumption implies that we can ensure convergence to any reachable target trajectory $(z^p,v,w)$ using the feedback $\kappa$.
Similar incremental stabilizability conditions have been considered in~\cite{kohler2018nonlinear} to study the stability of trajectory tracking MPC schemes without terminal ingredients.
Corresponding functions $V_s,\kappa$ can be designed offline  using control contraction metrics~\cite{manchester2017control} or a quasi-LPV design~\cite{koehler2020quasi,koelewijn2019linear}.
\begin{proposition}
\label{prop:output_regulation}
Let Assumptions~\ref{ass:regulator}--\ref{ass:increm_stab} hold. 
There exists a constant  $\delta>0$, such that for any initial condition $x^p_0$ satisfying $V_{s}(x^p_0,\pi_x(w_0),w_0)\leq \delta$, the feedback $u=\kappa(x^p,\pi_x(w),w,\pi_u(w))$ ensures (uniform) exponential stability of the origin $e=0$ and constraint satisfacton, i.e., $(x^p_t,u_t)\in\mathcal{Z}^p$ for all $t\geq 0$. 
\end{proposition}
\begin{proof}
Consider $(z^p_t,u_t)=(\pi_x(w_t),\pi_u(w_t))$ and note that $e=x^p-\pi_x(w)=x^p-z^p$.
Due to Ass.~\ref{ass:regulator}, we know that  $(z^p_t,v_t)=(\pi_x(w_t),\pi_u(w_t))\in\text{int}(\mathcal{Z}^p)$. 
For $\delta\leq \delta_{loc}$, applying Inequality~\eqref{eq:increm_a} repeatedly ensures $c_{s,l}\|e_t\|^2\leq V_{s}(x^p_t,z_p^t,w_t)\leq \rho^t V_{s}(x^p_0,z^p_0,w_0)\leq c_{s,u}\rho^t\|e_0\|^2$ and thus exponential stability of $e=0$.  
Given that $(z^p_t,v_t)\in\text{int}(\mathcal{Z}^p)$ (Ass.~\ref{ass:regulator}), there exists a constant $\epsilon_s>0$, such that $(x^p_t,u_t)\in\mathcal{Z}^p$ if  $\|(x_t^p-z^p_t,u_t-v_t)\|\leq \epsilon_s$. 
Due to continuity of $\kappa$ and $V_{s}$,  there exists a constant $\delta\in(0,\delta_{loc}]$, such that $V_{s}(x^p,z^p,w)\leq \delta$ implies $(x^p,\kappa(x^p,z^p,w,v))\in\mathcal{Z}^p$, which finishes the proof.  
\end{proof}

Essentially, this result is similar to the output regulation results in~\cite{pavlov2006uniform}, which use $\pi_u$ as a feedforward input in combination with convergent dynamics, which is closely related to incrementally stable dynamics, compare~\cite{tran2019convergence}.
We point out that Proposition~\ref{prop:output_regulation} only provides a \textit{local} solution to the constrained output regulation problem and requires knowledge of $\pi_x$, $\pi_u$, the solution to the regulator equations~\eqref{eq:reg}.
Both of these restrictions will be relaxed in the proposed MPC approach.
%

\subsection{Output regulation MPC}
\label{sec:reg_2}
In the following, we present the proposed output regulation MPC scheme.  
To simplify the notation, we introduce the overall state by $x_t:=(x^p_t,w_t)\in\mathbb{X}=\mathbb{R}^{n_p}\times\mathbb{W}\subseteq\mathbb{R}^n$ consisting of the plant state $x^p$ and the state of the exosystem $w$.  
The constraints and dynamics for this overall system are given by $\mathcal{Z}=\{(x^p,w,u)\in\mathbb{X}\times\mathbb{U}|~(x^p,u)\in\mathcal{Z}^p\}$ and  $f(x,u):=(f^p(x^p,u,w),s(w))$.
We consider the following stage cost 
\begin{align}
\label{eq:ell_reg}
\ell(x,u):=\|h(x^p,u,w)\|_Q^2+\|u-\pi_u(w)\|_R^2, 
\end{align}
 which penalizes the  tracking error $y$ and contains an input regularization w.r.t. the feedforward input $\pi_u(w)$ using positive definite matrices $Q=Q^\top\in\mathbb{R}^{p\times p}$, $R=R^\top\in\mathbb{R}^{m\times m}$. 
 The stage cost $\ell$~\eqref{eq:ell_reg} requires the knowledge of the input feedforward $\pi_u(w)$. 
This may be rather restrictive and is only used for some preliminary analysis in Theorem~\ref{thm:MPC} based on existing methods from~\cite{grimm2005model}.
Sections~\ref{sec:minPhase} and \ref{sec:input} will be devoted to dropping this requirement by assuming that the system is minimum phase or suitable changing the MPC formulation in~\eqref{eq:MPC}.
The MPC optimization problem is given as
\begin{subequations}
\label{eq:MPC}
\begin{align}
\label{eq:MPC_1}
V_N(x_t):=\inf_{u_{\cdot|t}\in\mathbb{U}^N}& J_N(x_{\cdot|t},u_{\cdot|t}):=\sum_{k=0}^{N-1} \ell(x_{k|t},u_{k|t}),\\
\text{s.t. }&x^p_{0|t}=x_t^p,~w_{0|t}=w_t,\\
&w_{k+1|t}=s(w_{k|t}),\\
&x^p_{k+1|t}=f^p(x^p_{k|t},u_{k|t},w_{k|t}),\\
&(x^p_{k|t},u_{k|t})\in\mathcal{Z}^p,\\
&k=0,\dots,N-1,\nonumber
\end{align}
\end{subequations}
with a prediction horizon $N\in\mathbb{N}$. 
For simplicity we assume $\mathbb{U}$ compact. 
The solution\footnote{%
A minimizer exists, since $f^p,\ell,h$ are continuous and $\mathbb{U}$ is compact. 
If the minimizer is not unique, an arbitrary minimizer can be chosen.}
 to this optimization problem is the value function $V_N$ and optimal state and input trajectories $(x_{\cdot|t}^*,u_{\cdot|t}^*)$ with $x_{k|t}=(x^p_{k|t},w_{k|t})$.
The resulting closed-loop system is given by
\begin{align}
\label{eq:close}
u_t=u^*_{0|t},\quad x_{t+1}=f(x_t,u^*_{0|t})=x^*_{1|t},~t\geq 0. 
\end{align}
In order to solve the optimization problem~\eqref{eq:MPC} with the stage cost~\eqref{eq:ell_reg}, we need to be able to predict both the plant state $x^p$ and the exosystem state $w$, and thus we initially assume that the overall state $x=(x^p,w)$ can be measured and an accurate prediction model is available.
The extension to the classical error feedback setup~\cite[Ch.~8]{isidori2013nonlinear}, where only $y$ can be measured will be discussed in Remark~\ref{rk:error_feedback}.
We note that the computational complexity of the nonlinear MPC problem~\eqref{eq:MPC} and the optimization problems presented in Sections~\ref{sec:minPhase} and \ref{sec:input} are comparable to a standard nonlinear MPC without terminal ingredients~\cite{grune2012nmpc}, since $w_{\cdot|t}$ can be predicted independent of $u_{\cdot|t}$ and no additional decision variables or constraints are included in the optimization problem.
 
\subsection{Theoretical analysis}
\label{sec:reg_3}
We first extend the general analysis in~\cite{grimm2005model} using stabilizability and detectability conditions in Theorem~\ref{thm:MPC}. 
In Corollary~\ref{corol:main}, we show that the MPC scheme solves the output regulation problem using incremental properties of the plant.
\subsubsection{MPC  - semidefinite costs and set stabilization}
We consider the following continuous  state measure to be minimized
\begin{align}
\label{eq:sigma}
\sigma(x):=\|x^p-\pi_x(w)\|^2=\|e\|^2.
\end{align} 
Achieving $\sigma(x)=0$  is equivalent to driving the overall state to the regulator manifold $\mathcal{A}:=\{x\in\mathbb{X}|~x^p=\pi_x(w)\}$. 
We consider the following stabilizability and detectability conditions, similar to~\cite[SA3/4]{grimm2005model}. 
\begin{assumption}
\label{ass:stab}
(Local cost stabilizability)
There exist constants $\gamma_s,\delta_s>0$, such that for any $x\in\mathbb{X}_{\delta}:=\{x\in\mathbb{X}|~\sigma(x)\leq \delta_s\}$, Problem~\eqref{eq:MPC} is feasible and the value function satisfies 
\begin{align}
\label{eq:stab_grimm}
V_N(x)\leq \gamma_s \sigma(x),\quad \forall N\in\mathbb{N}.
\end{align}
\end{assumption} 
Compared to~\cite{grimm2005model}, Assumption~\ref{ass:stab} only assumes \textit{local} stabilizability, which is significantly easier to verify and also applicable to unstable systems and/or in the presence of state constraints, which are not control invariant.  
In case $\sigma(x)=\|x\|^2$ (no exosystem), Assumption~\ref{ass:stab} corresponds to the local \textit{exponential cost controllability} condition in~\cite[Ass.~1]{boccia2014stability}, which is less restrictive than the global condition used in~\cite[Ass.~3.5]{grune2012nmpc}. Similar local bounds on the value function are used in tracking MPC  with and without terminal constraints in~\cite[Ass.~2]{kohler2020nonlinear} and \cite[Prop.~2]{kohler2018nonlinear}, respectively.
\begin{assumption}
\label{ass:detect}
(Cost detectability) There exists a function $W:\mathbb{X}\rightarrow\mathbb{R}_{\geq 0}$ and constants $\gamma_o$, $\epsilon_o>0$, such that 
\begin{subequations}
\label{eq:detect_grimm}
\begin{align}
\label{eq:detect_grimm_1}
W(x)\leq &\gamma_o \sigma(x),\\
\label{eq:detect_grimm_2}
W(f(x,u))-W(x)\leq &-\epsilon_o\sigma(x)+\ell(x,u),  
\end{align}
\end{subequations}
for any $(x,u)\in\mathcal{Z}$.  
\end{assumption} 
\begin{remark}
\label{rk:dissip_detect}
Assumption~\ref{ass:detect} is a special case of the strict dissipativity condition typically used in economic MPC~\cite{muller2021dissipativity}. 
The main difference is that $W$ satisfies the upper bound~\eqref{eq:detect_grimm_1}, while in economic MPC only boundedness (from below) of $W$ is assumed, compare~\cite{hoger2019relation}. 
This small technical difference is the main reason that the analysis of economic MPC schemes and the resulting performance bounds are significantly more conservative, compare~\cite{grune2013economic}. 
Note that Assumption~\ref{ass:detect} is trivially satisfied with $W=0$ if $\ell(x,u)\geq \sigma(x)$, which is the standard case studied in the MPC literature, compare e.g.~\cite{grune2012nmpc,kohler2018nonlinear}. 
\end{remark}
The following theorem combines ideas from~\cite[Thm.~1--2]{grimm2005model} to deal with positive semidefinite  stage costs $\ell$ using a detectability condition (Ass.~\ref{ass:detect}) with the methods in~\cite[Thm.~1--2]{kohler2018nonlinear} to consider less restrictive \textit{local} stabilizability conditions (Ass.~\ref{ass:stab}). 
\begin{theorem}
\label{thm:MPC}
Let Assumptions~\ref{ass:stab}--\ref{ass:detect} hold. 
For any constant $\overline{Y}>0$,  there exist constants $N_{\overline{Y}}$, $\gamma_{\overline{Y}}>0$, such that for $N>N_{\overline{Y}}$ and  initial condition $x_0\in \mathbb{X}_{\overline{Y}}:=\{x\in\mathbb{X}|~V_N(x)+W(x)\leq \overline{Y}\}$, the MPC problem~\eqref{eq:MPC} is recursively feasible, the constraints are satisfied, and the function $Y_N:=V_N+W$ satisfies 
\begin{subequations}
\label{eq:Lyap_MPC}
\begin{align}
\label{eq:Lyap_1}
\epsilon_o\sigma(x_t)\leq Y_N(x_t) \leq &\gamma_{\overline{Y}}\sigma(x_t),\\
\label{eq:Lyap_2}
Y_N(f(x_t,u_t))-Y_N(x_t)\leq& -\alpha_N\epsilon_o\sigma(x_t),
\end{align}
for the resulting closed loop with
\begin{align}
\label{eq:def_alpha}
\alpha_N:=1-\dfrac{\gamma_s\gamma_{\overline{Y}}}{\epsilon_o^2(N-1)}>0.
\end{align}
Furthermore, the closed loop satisfies the following  transient performance bound   
\begin{align}
\label{eq:perform}
\alpha_N\sum_{t=0}^{T-1} \ell(x_t,u_t)\leq Y_N(x_0)\leq  \overline{Y},\quad \forall T\in\mathbb{N}. 
\end{align}
\end{subequations}
\end{theorem}
\begin{proof}
\begin{subequations}
\label{eq:Lyap_MPC_proof}
We first show  inequalities~\eqref{eq:Lyap_1}--\eqref{eq:Lyap_2} for all $x_t\in\mathbb{X}_{\overline{Y}}$, and then the performance bound~\eqref{eq:perform}.
 Abbreviate $\ell_{k}=\ell(x^*_{k|t},u^*_{k|t})$. \\
\textbf{Part I. }
The lower bound in~\eqref{eq:Lyap_1} follows with $\ell_0\geq 0$, and
\begin{align*}
&Y_N(x_t)\stackrel{\eqref{eq:MPC_1}}\geq \ell_{0}+W(x_t)
\stackrel{\eqref{eq:detect_grimm_2}}{\geq}\epsilon_o\sigma(x_t)+W(x^*_{1|t})\geq \epsilon_o\sigma(x_t).  
\end{align*}
For any $x_t\in\mathbb{X}_{\delta}$, we directly obtain the bound $Y_N(x_t)\leq (\gamma_s+\gamma_o)\sigma(x_t)$ using~\eqref{eq:stab_grimm}, \eqref{eq:detect_grimm_1}. 
The upper bound in~\eqref{eq:Lyap_1} holds with  $\gamma_{\overline{Y}}:=\max\left\{\gamma_s+\gamma_o,\frac{\overline{Y}}{\delta_s}\right\}$ using this bound,  $x_t\in\mathbb{X}_{\overline{Y}}$ and a case distinction whether or not $x_t\in\mathbb{X}_{\delta}$, similar to~\cite[Thm.~2]{kohler2018nonlinear}, \cite{boccia2014stability}. \\
\textbf{Part II. }
The detectability condition (Ass.~\ref{ass:detect}) implies
\begin{align}
\label{eq:rotated}
&W(x^*_{N|t})-W(x_t)=\sum_{k=0}^{N-1}W(x^*_{k+1|t})-W(x^*_{k|t})\nonumber\\
\stackrel{\eqref{eq:detect_grimm_2}}{\leq}& \sum_{k=0}^{N-1}-\epsilon_o \sigma(x^*_{k|t})+\sum_{k=0}^{N-1}\ell_{k}.
\end{align}
Using $W(x^*_{N|t})\geq 0$, and $x_t\in\mathbb{X}_{\overline{Y}}$, we arrive at
\begin{align}
\label{eq:rotated_2}
\epsilon_o\sum_{k=0}^{N-1} \sigma(x^*_{k|t})\stackrel{\eqref{eq:rotated}}{\leq} Y_N(x_t)\stackrel{\eqref{eq:Lyap_1}}{\leq} \min\{\overline{Y},\gamma_{\overline{Y}} \sigma(x_t)\}.
\end{align}
Thus, there exists a $k_x\in\{1,\dots,N-1\}$, such that
\begin{align}
\label{eq:turnpike}
\sigma(x^*_{k_x|t})\leq \dfrac{\min\{\overline{Y},\gamma_{\overline{Y}}\sigma (x_t)\}}{\epsilon_o(N-1)}.
\end{align}
Given $N\geq N_0:=1+\frac{\overline{Y}}{\delta_s\epsilon_o}$, this implies $x^*_{k_x|t}\in\mathbb{X}_\delta$. 
Thus, Assumption~\ref{ass:stab} ensures that starting at $x_{k_x|t}^*$ there exists a feasible state and input trajectory satisfying the bound~\eqref{eq:stab_grimm}, which implies
\begin{align}
\label{eq:Value_dec}
&V_N(x_{t+1})+\ell_0\leq \sum_{k=0} ^{k_x-1}\ell_{k}+V_{N-k_x+1}(x^*_{k_x|t})\\
\stackrel{\eqref{eq:stab_grimm}}{\leq} &V_N(x_t)+ \gamma_s\sigma(x^*_{k_x|t})\stackrel{\eqref{eq:turnpike}}{\leq} V_N(x_t)+\frac{\gamma_s\gamma_{\overline{Y}}}{\epsilon_o(N-1)}\sigma(x_t).\nonumber
\end{align}
Combining \eqref{eq:Value_dec} and \eqref{eq:detect_grimm_2}, the function $Y_N$ satisfies~\eqref{eq:Lyap_2}. 
Furthermore,  $\alpha_N>0$ follows from~\eqref{eq:def_alpha},  using $N> N_1:=1+\gamma_s\gamma_{\overline{Y}}/\epsilon_o^2$.
All the arguments hold with $N> N_{\overline{Y}}:=\max\{N_0,N_1\}$. 
In addition, Inequality~\eqref{eq:Lyap_2} ensures $Y_N$ is nonincreasing and thus $x_t\in\mathbb{X}_{\overline{Y}}$ for all $t\geq 0$. \\
\textbf{Part III. }
Using $Y_N=V_N+W \geq V_N\geq 0$, $\alpha_N\leq 1$ and the following inequality in a telescopic sum
\begin{align}
\label{eq:perform_intermediate}
&Y_N(x_{t+1})-Y_N(x_t)-\alpha_N(W(x_{t+1})-W(x_t))\nonumber\\
=&(1-\alpha_N)(Y_N(x_{t+1})-Y_N(x_t))+\alpha_N(V_N(x_{t+1})-V_N(x_t)) \nonumber\\
&\stackrel{\eqref{eq:Lyap_2},\eqref{eq:Value_dec}}{\leq}-\alpha_N\ell(x_t,u_t), 
\end{align}
implies the performance bound~\eqref{eq:perform}. 
\end{subequations}
 \end{proof}
Inequalities~\eqref{eq:Lyap_1}--\eqref{eq:Lyap_2} directly imply $\lim_{t\rightarrow\infty}\sigma(x_t)=0$ and hence asymptotic convergence. 
Compared to~\cite[Thm.~1]{grimm2005model}, Theorem~\ref{thm:MPC} considers a less restrictive (local) stabilizability condition to show stability and provides a performance bound~\eqref{eq:perform} similar to the suboptimality estimates usually obtained in MPC without terminal constraints (with $\ell$ positive definite), compare e.g.~\cite{grune2012nmpc,reble2012unconstrained}. 
The intermediate bound~\eqref{eq:perform_intermediate} and the definition of $\alpha_N$ in \eqref{eq:def_alpha} imply that as $N\rightarrow\infty$, we recover $\sum_{t=0}^{\infty}\ell(x_t,u_t)\leq V_\infty(x)$, similar to standard results with $\ell$ positive definite~\cite{grune2012nmpc,reble2012unconstrained}. 
 
\begin{remark}
\label{rk:observable} (Improved bounds using observability)
The (quantitative) bounds in Theorem~\ref{thm:MPC} for $\alpha_N$, $N_{\overline{Y}}$ can be improved if the detectability condition (Ass.~\ref{ass:detect}) is strengend to a finite-time observability condition w.r.t. the stage cost $\ell$. 
In particular, the suboptimality bound $1-\alpha_N$~\eqref{eq:def_alpha} in Theorem~\ref{thm:MPC} decreases with $\gamma^2/N$ {(consider $\gamma\approx \gamma_s\approx \gamma_{\overline{Y}}$ in this discussion to keep in line with the notation in~\cite{grune2012nmpc})}, which is comparable to the bound in~\cite[Variant 1]{grune2012nmpc} and \cite{grimm2005model}. 
Under the assumption that $\ell$ is positive definite (Ass.~\ref{ass:detect} holds with $W=0$) the derivations in~\cite{tuna2006shorter}, \cite[Variant 2]{grune2012nmpc} and \cite[Thm.~2]{kohler2018nonlinear}, provide bounds where $1-\alpha_N$  decreases exponentially in $N$.
Conceptually similar bounds can also be recovered for the present setting with the detectability condition (Ass.~\ref{ass:detect}), given an additional finite-time observability condition. 
The corresponding theoretical details can be found in Appendix~\ref{app:obs}.
\end{remark}

\subsubsection{Incremental system properties}
In the following, we derive sufficient conditions for Assumptions~\ref{ass:stab}--\ref{ass:detect}, assuming the regulator equations admit a solution (Ass.~\ref{ass:regulator}).
\begin{proposition}
\label{prop:stab}
Let Assumptions~\ref{ass:regulator} and \ref{ass:increm_stab} hold and suppose that $h$ is locally Lipschitz continuous. 
Then Assumption~\ref{ass:stab} holds.  
\end{proposition}
\begin{proof}
In Proposition~\ref{prop:output_regulation} it was already shown that $u=\kappa(x^p,\pi_x(w),w,\pi_u(w))$ is a feasible control input for $V_{s}(x^p_0,z^p_0,w_0)\leq \delta$. Thus the optimization problem~\eqref{eq:MPC} is feasible with this candidate input for all $x_t\in\mathbb{X}_\delta$ with $\delta_s:=\delta/c_{s,l}$. 
It remains to show that the bound~\eqref{eq:stab_grimm} holds. 
Similar to~\cite[Prop.~2]{kohler2018nonlinear}, Inequalities~\eqref{eq:increm} and $\kappa$ being Lipschitz imply that there exists a constant $c>0$ such that 
\begin{align*}
\|(x^p_{k|t},u_{k|t})-(\pi_x(w_{k|t}),\pi_u(w_{k|t}))\|^2\leq c\rho_s^{k} \|x^p_t-\pi_x(w_t)\|^2.
\end{align*} 
Lipschitz continuity of $h$ with some Lipschitz constant $L_h$  and $h(\pi_x(w_{k|t}),\pi_u(w_{k|t}),w_{k|t})\stackrel{\eqref{eq:reg2}}{=}0$ imply
\begin{align*}
&V_N(x_t)\leq \sum_{k=0}^{N-1}\|h(x^p_{k|t},u_{k|t},w_{k|t})\|_{Q}^2+\|u_{k|t}-\pi_u(w_{k|t})\|_R^2\\
\stackrel{}{\leq}&  \max\{L_h^2\lambda_{\max}(Q),\lambda_{\max}(R)\}c \|x^p_t-\pi_x(w_t)\|^2\sum_{k=0}^{N-1}\rho_s^{k}\\
\leq &\underbrace{\dfrac{\max\{L_h^2\lambda_{\max}(Q),\lambda_{\max}(R)\}c }{1-\rho_s}}_{=:\gamma_s} \underbrace{\|x^p_t-\pi_x(w_t)\|^2}_{=\sigma(x_t)}.\nonumber &\qedhere
\end{align*} 
\end{proof}
Note that Ass.~\ref{ass:increm_stab} could be relaxed to only hold for $(z^p,v)=(\pi_x(w),\pi_u(w))$, which is less restrictive (compare convergent dynamics in~\cite{pavlov2006uniform}). 
However, the benefit of Ass.~\ref{ass:increm_stab} is that it can be verified without solving the regulator equations~\eqref{eq:reg}, which is one of the main motivations of the presented work.

The following assumption characterizes the detectability of the nonlinear system using the standard notion of incremental input-output-to-state stability (i-IOSS)~\cite{krichman2001input,cai2008input}.
\begin{assumption}
\label{ass:IOSS} (exponential i-IOSS)
There exist a continuous i-IOSS Lyapunov function $V_{o}:\mathbb{X}^p\times\mathbb{X}^p\times\mathbb{W}\rightarrow\mathbb{R}_{\geq 0}$ and
constants $c_{o,l}$, $c_{o,u}$, $c_{o,1}$, $c_{o,2}>0$, $\rho_o\in(0,1)$, such that for any  $(z^p,v,w,x^p,u)\in\mathcal{Z}^p\times\mathbb{W}\times\mathcal{Z}^p$, we have
\begin{subequations}
\label{eq:IOSS}
\begin{align}
\label{eq:IOSS_1}
&c_{o,l}\|x^p-z^p\|^2\leq V_{o}(x^p,z^p,w)\leq c_{o,u}\|x^p-z^p\|^2,\\
\label{eq:IOSS_2}
&V_{o}(f^p(x^p,u,w),f^p(z^p,v,w),s(w))-\rho_o V_{o}(x^p,z^p,w)\\
\leq &c_{o,1}\|u-v\|^2+c_{o,2}\|h(x^p,u,w)-h(z^p,v,w)\|^2.\nonumber
\end{align}
\end{subequations}
\end{assumption}
We note that i-IOSS is a necessary condition for the existance of a (robustly) stable observer~\cite{allan2020detect,knufer2020time} and hence not a too restrictive conditions for the output regulation problem.
A corresponding i-IOSS Lyapunov function $V_o$ can be computed using results for \textit{differential detectability}~\cite{sanfelice2012metric} or more generally results from \textit{incremental dissipativity}~\cite{verhoek2020convex}.
\begin{proposition}
\label{prop:detect}
Let Assumptions~\ref{ass:regulator} and \ref{ass:IOSS} hold.
Then Assumption~\ref{ass:detect} holds. 
\end{proposition}
\begin{proof}
Assumption~\ref{ass:regulator} implies $(z^p,v)=(\pi_x(w),\pi_u(w))\in\text{int}(\mathcal{Z}^p)$ and thus Assumption~\ref{ass:IOSS} ensures~\eqref{eq:IOSS}.
Consider $W(x)=c V_o(x^p,\pi_x(w),w)$ with $c=\min\{\frac{\lambda_{\max}(R)}{c_{o,1}},\frac{\lambda_{\max}(Q)}{c_{o,2}}\}>0$. 
The upper bound~\eqref{eq:detect_grimm_1} follows directly from~\eqref{eq:IOSS_1} with $\gamma_o=c\cdot c_{o,u}$  and $\sigma(x)=\|x^p-z^p\|^2$. 
Inequality~\eqref{eq:detect_grimm_2} holds  with $\epsilon_o=(1-\rho_o)\cdot c\cdot c_{o,l}>0$ using
\begin{align*}
&W(f(x,u))-W(x)\\
\stackrel{\eqref{eq:reg2},\eqref{eq:IOSS_2}}{\leq}&c\left( c_{o,1}\|u-\pi_u(w)\|^2+ c_{o,2}\|h(x^p,u,w)\|^2\right)\\
& -(1-\rho_o)c V_o(x^p,w)\\
\stackrel{\eqref{eq:IOSS_1}}{\leq}&\dfrac{c\cdot c_{o,1}}{\lambda_{\max}(R)}\|u-\pi_u(w)\|_R^2+\dfrac{c\cdot c_{0,2}}{\lambda_{\max}(Q)}\|h(x^p,u,w)\|_Q^2\\
& -(1-\rho_o)c\cdot  c_{o,l}\|x^p-\pi_x(w)\|^2\\
\stackrel{\eqref{eq:ell_reg},\eqref{eq:sigma}}{\leq} &-\epsilon_o\sigma(x)+\ell(x,u).&\qedhere
\end{align*}
\end{proof}
\begin{remark}
\label{rk:ell_PDF}
One may be tempted to consider $\sigma(x)=\min_{u\in\mathbb{U}}\ell(x,u)$, which satisfies Assumption~\ref{ass:detect} with $W=0$, $\epsilon_o=1$. 
However, in this case Inequality~\eqref{eq:stab_grimm} in Assumption~\ref{ass:stab} is typically only satisfied if the stage cost $\ell$ is positive definite w.r.t. $(x^p,u)=(\pi_x(w),\pi_w(w))$, which is quite restrictive. 
This special case of tracking a given state-input reference trajectory is treated in~\cite[Thm.~2]{kohler2018nonlinear}.\footnote{%
We note that in~\cite[Thm.~4]{kohler2018nonlinear}, also the case of unreachable trajectories is treated, i.e., when Assumption~\ref{ass:regulator} does not hold.} 
\end{remark}

Given Prop.~\ref{prop:stab}--\ref{prop:detect}, we can now recast Thm.~\ref{thm:MPC} using intuitive assumptions on the inherent system properties of the plant. 
\begin{corollary}
\label{corol:main}
Let Assumptions~\ref{ass:regulator}, \ref{ass:increm_stab}, \ref{ass:IOSS}  hold. 
Suppose further that $\pi_x$ from Ass.~\ref{ass:regulator} and $h$ are Lipschitz continuous. 
For any constant $\overline{Y}>0$,  there exists a constant $N_{\overline{Y}}$, such that for $N>N_{\overline{Y}}$ and  initial condition $x_0\in \mathbb{X}_{\overline{Y}}$
, the MPC problem~\eqref{eq:MPC} is recursively feasible, the constraints are satisfied, and the regulator manifold $\mathcal{A}$ is exponentially stable for the resulting closed loop.
\end{corollary}
\begin{proof}
First, note that Propositions~\ref{prop:stab}--\ref{prop:detect} ensure that Assumptions~\ref{ass:stab}--\ref{ass:detect} hold. 
Thus, for $N>N_{\overline{Y}}$, with $N_{\overline{Y}}$ from Thm.~\ref{thm:MPC}, the bounds~\eqref{eq:Lyap_MPC} hold with $\alpha_N>0$. 
Define the point-to-set distance $\|x\|_{\mathcal{A}}:=\inf_{s\in\mathcal{A}}\|x-s\|$. 
In the following, we show that there exists a constant $c_\pi>0$, such that $c_{\pi}\sigma(x)\leq \|x\|^2_{\mathcal{A}}\leq \sigma(x)$, which in combination with~\eqref{eq:Lyap_MPC} ensures exponential stability of $\mathcal{A}$ using standard Lyapunov arguments. 
For given $(x^p,w)$, denote some minimizer by $\tilde{w}:=\arg\min_{\tilde{w}\in\mathbb{W}}\|(\pi_x(\tilde{w}),\tilde{w})-(x^p,w)\|$. 
Given that $\pi_x$ is Lipschitz continuous with Lipschitz constant $L_\pi$, we have 
\begin{align*}
\sigma(x) =&\|x^p-\pi_x(w)\|^2\\
\leq& 2(\|x^p-\pi_x(\tilde{w})\|^2+\|\pi_x(\tilde{w})-\pi_x(w)\|^2)\\
\leq &2\max\{L_\pi^2,1\}\|(x^p,w)-(\pi_x(\tilde{w}),\tilde{w})\|^2=:1/c_\pi \|x\|_{\mathcal{A}}^2,
\end{align*}
where the first inequality uses $\|a+b\|^2\leq 2(\|a\|^2+\|b\|^2)$ for any $a,b\in\mathbb{R}^{n_p}$. 
 Furthermore, 
\begin{align*}
\|x\|_{\mathcal{A}}=&\|(\pi_x(\tilde{w}),\tilde{w})-(x^p,w)\|
{\leq}\|x^p-\pi_x(w)\|\stackrel{\eqref{eq:sigma}}{=} \sqrt{\sigma(x)},
\end{align*}
which finishes the proof. 
\end{proof}
Overall, this result implies that the proposed MPC scheme solves the nonlinear constrained regulation problem if:
\begin{enumerate}[label=(\alph*)]
\item The regulator problem is (strictly) feasible (Ass.~\ref{ass:regulator}),
\label{enum:cond_a}
\item The plant is incrementally stabilizable (Ass.~\ref{ass:increm_stab}) and detectable (i-IOSS, Ass.~\ref{ass:IOSS}).
\label{enum:cond_b}
\end{enumerate}
The main practical restriction of the proposed formulation is the fact that we need the feedforward $\pi_u$ to implement the input regularization in the stage cost~\eqref{eq:ell_reg}. 
This shortcoming will be removed in Sections~\ref{sec:minPhase} and \ref{sec:input} using additional conditions on the zero dynamics or a modified MPC formulation, respectively.
Thus, with these formulations the solution to the regulator equations $\pi_x,\pi_u$ is \emph{not} needed for the implementation and instead the closed loop will ``find'' the regulator manifold, which is the crucial advantage of the proposed formulation compared to, e.g., classical trajectory stabilization (Prop.~\ref{prop:output_regulation}, \cite{isidori1990output,castillo1993nonlinear,pavlov2006uniform}).
In addition, compared to Prop.~\ref{prop:output_regulation}, the proposed MPC scheme yields a larger region of attraction (despite the presence of hard constraints).

 \begin{remark}
 \label{rk:error_feedback}
(Error feedback and robustness)
The output regulation problem is classically posed without state measurements and solved using a dynamic error feedback, compare~\cite{isidori1990output} and \cite{byrnes2003limit,priscoli2009dissipativity}.  
In the present paper, we restrict ourselves to the nominal case of exact state measurements, but the proposed MPC framework can be naturally extended to the error feedback case using an observer and tools from output-feedback MPC~\cite{findeisen2003state}.
The theoretical details can be found in Appendix~\ref{app:error_feedback}, where we show finite-gain $\mathcal{L}_2$ stability in the presence of noisy output measurements, given some simplifying assumptions (mainly no state constraints).
 \end{remark}

\begin{remark}
\label{rk:IO_model} (Practical applicability: tracking \& I/O costs)
By addressing the output regulation problem, the provided framework is applicable to dynamic trajectory tracking, disturbance rejection and output feedback~\cite{isidori1990output,castillo1993nonlinear,pavlov2006uniform,byrnes2012output}.
In particular, any time-varying reference trajectory and disturbance signal can be viewed as an output of the exosystem (e.g., by treating the time as a state of the exosystem). 
The main restriction is the fact that the reference and disturbance signal can be predicted and the regulation problem is feasible (Ass.~\ref{ass:regulator}). 
For comparison, in the related literature on trajectory tracking MPC,  a given state reference trajectory is considered and also a suitable terminal cost (and terminal region) need to be constructed offline~\cite{koehler2020quasi,faulwasser2015nonlinear}. 
Using Theorem~\ref{thm:MPC}/Corollary~\ref{corol:main}, the need for constructing such terminal ingredients or determining the corresponding state trajectory can be dropped by choosing a  sufficiently  large $N$.

The stability results in this paper are also important for MPC with input-output stage costs $\ell=\|y\|_Q^2+\|u\|_R^2$. 
This is of relevance for input-output models resulting from system identification, e.g.  input-output LPV systems~\cite{abbas2018improved,cisneros2019stabilizing}. 
Similarly, input-output stage costs $\ell$ and avoidance of terminal ingredients or state references also appears in data-driven MPC, where (linear) models are implicitly represented using data~\cite{berberich2019data,coulson2019regularized}.
A corresponding extension of the results in this paper to derive a robust data-driven MPC can be found in~\cite{bongard2021robust}.
\end{remark}

\section{Minimum phase systems}
\label{sec:minPhase}
In this section, we show that in case of minimum phase systems, the proposed output regulation MPC (Sec.~\ref{sec:reg}) ensures stability, even without any input regularization. 
Section~\ref{sec:minPhase_1} introduces preliminaries regarding relative degree and the zero dynamics.
Section~\ref{sec:minPhase_2} shows that the minimum phase property implies the detectability condition (Ass.~\ref{ass:detect}) for a look-ahead stage cost $\ell_d$. 
Section~\ref{sec:minPhase_3} shows that the MPC formulation~\eqref{eq:MPC} in Section~\ref{sec:reg} also yields the desired closed-loop properties  without input regularization.
Some discussion can be found in Section~\ref{sec:minPhase_4}.

\subsection{Relative degree -  Byrnes-Isidori normal form}
\label{sec:minPhase_1}
For simplicity of exposition, we consider a single-input-single-output (SISO) system without a direct feed through term, i.e. $m=p=1$ and $h(x^p,u,w)=h(x)$. 
We assume that the system has no direct feed through to keep in line with the setup in the relevant literature, compare~\cite{monaco1987minimum}.
The multi-input-multi-output (MIMO) case is discussed in Remark~\ref{rk:MIMO_min_phase} below. 
For ease of notation, Assumptions~\ref{ass:BINF}--\ref{ass:minPhase} below regarding the relative degree and the zero dynamics will be posed globally.
\subsubsection*{Relative degree}
We consider the case, where the system has a well defined relative degree $d\in\mathbb{N}$, which is characterized using the Byrnes-Isidori normal form (similar to~\cite[Prop.~2.1]{monaco1987minimum}).
\begin{assumption}
\label{ass:BINF}
(Byrnes-Isidori normal form)
There exist a Lipschitz continuous  function $\Phi:\mathbb{X}\rightarrow\mathbb{R}^{n_p}$ and a constant $d\in\mathbb{N}$, such that the state $\zeta=(z,\eta)=:\Phi(x^p,w)=(\Phi_1(x^p,w),\Phi_2(x^p,w))$, $\eta\in\mathbb{R}^{n_p-d-1}$, $z=(z^1,\dots,z^{d+1})\in\mathbb{R}^{d+1}$ is subject to the following dynamics for all $t\geq 0$:
\begin{subequations}
\label{eq:BINF}
\begin{align}
\label{eq:BINF_1}
z^k_{t+1}=&z^{k+1}_t,\quad k=1,\dots,d,\\
\label{eq:BINF_2}
z^{d+1}_{t+1}=&F_1(\zeta_t,w_t,u_t),\\
\label{eq:BINF_3}
\eta_{t+1}=&F_2(\zeta_t,w_t,u_t),\\
\label{eq:BINF_4}
y_t=&z^1_t,
\end{align}
\end{subequations}
with $w_t$ according to~\eqref{eq:sys_exo} and  Lipschitz continuous maps $F_1,F_2$. 
Furthermore, there exists a Lipschitz continuous function $\tilde{\Phi}:\mathbb{R}^{n_p}\times\mathbb{W}\rightarrow\mathbb{R}^{n_p}$ satisfying
$\tilde{\Phi}(\Phi(x^p,w),w)=x^p$ for all $w\in\mathbb{W}$. \\
\end{assumption}
In the following, we abbreviate the dynamics~\eqref{eq:BINF_1}--\eqref{eq:BINF_3}  by $\zeta_{t+1}=:F(\zeta_t,w_t,u_t)$. 
Assumption~\ref{ass:BINF} ensures that the system can be transformed into the Byrnes-Isidori normal form (BINF). 
Lipschitz continuity of the inverse function $\tilde{\Phi}$ ensures that stability of the original plant $x^p$ can be studied based on the transformed state $\zeta$. 
With the presentation~\eqref{eq:BINF} we directly have $(y_t,...y_{t+d})=\Phi_1(x_t)$, i.e., in the next $d$ time steps the output $y$ cannot be influenced by the input $u$.
We have a well-defined relative degree if $\frac{\partial F_1}{\partial u}\neq 0$ for all $(w,\zeta,u)\in\mathbb{W}\times\mathbb{R}^{n_p+m}$, i.e., the input $u_t$ can influence the output $y_{t+d+1}$, which will be ensured through Assumption~\ref{ass:zero} below. 
We point out that in Assumption~\ref{ass:BINF} (and in Ass.~\ref{ass:minPhase} below) we only considered the BINF for the plant state $x^p$, but not the overall state $x=(x^p,w)$. 
In particular, the zero dynamics of $x$ contain the dynamics in $w$, which are in general not contractive.

\begin{assumption}
\label{ass:zero}
(Well-defined zero dynamics)
There exists a continuous control law $\tilde{\alpha}:\mathbb{R}^{n_p}\times\mathbb{W}\rightarrow\mathbb{R}^m$ and constants $c_{h_1},c_{h_2}>0$, such that 
\begin{align}
\label{eq:zero_dyn}
c_{h_1}|\Delta u|\leq |F_1(\zeta,w,\tilde{\alpha}(\zeta,w)+\Delta u)|\leq c_{h_2}|\Delta u|,
\end{align}
for all $\zeta\in\mathbb{R}^{n_p}$, $w\in\mathbb{W}$, $\Delta u\in\mathbb{R}^m$.
\end{assumption}
Consider the set $L_{D}=\{x\in\mathbb{R}^n|~\Phi_1(x)=0\}$, for which it holds that $y_t=0$ for $t=0,\dots,d$ if $x_0\in L_D$, as in~\cite{monaco1988zero}.
Condition~\eqref{eq:zero_dyn} ensures that there exists a \textit{unique} feedback law $\alpha(x):=\tilde{\alpha}(\Phi(x),w)$, such that the zero output manifold $L_D$ is positively
 invariant, which ensures that the system has a well-defined zero dynamics, compare~\cite[Sec.~V]{byrnes2003limit}.
The requirement of a \textit{unique} control law is relevant for well posedness of the zero dynamics and also the reason we restrict ourselves to SISO (or square MIMO) systems.

\subsection{Minimum phase systems and detectability}
\label{sec:minPhase_2}
The following assumption ensures that the system is minimum phase, i.e., the zero dynamics are asymptotically stable, using an ISS Lyapunov function.
\begin{assumption}
\label{ass:minPhase}
(Minimum phase)
There exist constants $\tilde{c}_{o,l},\tilde{c}_{o,u}>0$, $\rho_\eta\in(0,1)$ and an ISS Lyapunov function $V_\eta:\mathbb{R}^{n_p-d-1}\times\mathbb{W}$, such that for all  $(w,z,\eta,u)\in\mathbb{W}\times\mathbb{R}^{n_p+m}\rightarrow\mathbb{R}_{\geq 0}$ we have
\begin{subequations}
\label{eq:min_phase}
\begin{align}
\label{eq:min_phase_1}
&\tilde{c}_{o,l}\|\eta-\tilde{\eta}_w\|^2\leq V_\eta(\eta,w)\leq \tilde{c}_{o,u}\|\eta-\tilde{\eta}_w\|^2\\
\label{eq:min_phase_2}
&V_\eta(F_2(\zeta,w,u),s(w))\leq \rho_\eta V_\eta(\eta,w)+\|z\|^2+(u-\tilde{\alpha}(\zeta,w))^2,
\end{align}
\end{subequations}
with $\tilde{\eta}_w=\Phi_2(\pi_x(w),w)$,  $\zeta=(z,\eta)$.  
\end{assumption}
Given a system with exosystem state $w$ and consistently zero output ($z\equiv 0$, $u\equiv \tilde{\alpha}$,~cf. Ass.~\ref{ass:BINF}--\ref{ass:zero}), the state $\eta$ exponentially converges to $\tilde{\eta}_w$, which corresponds to the ``stationary'' value of $\eta$  for $(x^p,u)=(\pi_x(w),\pi_u(w))$. 
Furthermore, the dynamics of $\eta$ with $z=0,u=\tilde{\alpha}$ are a diffeomorphic copy of the dynamics of $f$ on $L_D$ and thus Assumption~\ref{ass:minPhase} characterizes the stability of the zero dynamics, i.e., Inequalities~\eqref{eq:min_phase} ensure the minimum phase property.
We point out that in~\cite{liberzon2002output}, the \textit{strongly} minimum phase property has been characterized using the notion of \textit{Output-input stability}, which is similar to the considered ISS characterization, compare~\cite[Example~2]{liberzon2002output}.

The following proposition shows that the minimum phase property guarantees the detectability condition (Ass.~\ref{ass:detect}) with a look-ahead stage cost $\ell_d$, similar to Proposition~\ref{prop:detect}.
\begin{proposition}
\label{prop:minphase_detect}
Let Assumptions~\ref{ass:regulator} and \ref{ass:BINF}--\ref{ass:minPhase} hold. 
Then Assumption~\ref{ass:detect} holds with the look-ahead stage cost $\ell_d(x,u):=h^2(x)+F_1^2(\Phi(x),w,u)$.
\end{proposition}
\begin{proof}
\begin{subequations}
\label{eq:minphase_detect}
Assumptions~\ref{ass:zero}--\ref{ass:minPhase} directly imply
\begin{align}
\label{eq:minphase_detect_1}
&V_\eta(F_2(\zeta,w,u),s(w))-V_\eta(\eta,w)\\
\stackrel{\eqref{eq:min_phase}}{\leq}& -(1-\rho_\eta)\cdot \tilde{c}_{o,l}\|\eta-\tilde{\eta}^w\|^2+(u-\tilde{\alpha}(\zeta,w))^2+\|z\|^2\nonumber\\
\stackrel{\eqref{eq:zero_dyn}}{\leq}&-(1-\rho_\eta)\cdot \tilde{c}_{o,l}\|\eta-\tilde{\eta}^w\|^2+\dfrac{F^2_1(\zeta,w,u)}{c_{h_1}}+\|z\|^2.\nonumber
\end{align} 
Note that the dynamics in $z$~\eqref{eq:BINF_1}--\eqref{eq:BINF_2} correspond to an FIR-filter with input $F_1$ and output $y$, which is hence detectable.
Thus, there exists a quadratic IOSS Lyapunov function with $P_z=P_z^\top\succ 0$ and some constants $\tilde{c}_{o,1},\tilde{c}_{o,2}>0$ satisfying 
\begin{align}
\label{eq:minphase_detect_2}
&\|z_{t+1}\|_{P_z}^2-\|z_t\|_{P_z}^2\nonumber\\
\leq& -\|z_t\|^2+\tilde{c}_{o,1}F_1^2(\zeta_t,w_t,u_t)+\tilde{c}_{o,2} y_t^2,
\end{align}
compare~\cite{cai2008input}.
The function $\tilde{W}(z,\eta,w):=c_1V_\eta(\eta,w)+c_2\|z\|_{P_z}^2$ with  $c_2:=\frac{1}{\max\{\tilde{c}_{o,2},2\tilde{c}_{o,1}\}}>0$, $c_1:=\frac{\min\{c_{h_1},c_2\}}{2}>0$ satisfies
\begin{align*}
&\tilde{W}(\zeta_{t+1},w_{t+1})-\tilde{W}(\zeta_t,w_t)\\
\stackrel{\eqref{eq:minphase_detect_1}\text{--}\eqref{eq:minphase_detect_2}}{\leq} &
-\tilde{\epsilon}_o(\|z_t\|^2+\|\eta_t-\tilde{\eta}^w_t\|^2)+F^2_1(\zeta_t,w_t,u_t)+y_t^2\\
\leq& \ell_d(x_t,u_t)-\epsilon_o\|x^p_t-\pi_x(w_t)\|^2,
\end{align*}
with $\tilde{\epsilon}_o:=\min\{\frac{c_2}{2},c_1(1-\rho_\eta)\tilde{c}_{o,l}\}>0$ and $\epsilon_o=\tilde{\epsilon}_o/L_{\tilde{\Phi}}^2$, where $L_{\tilde{\Phi}}$ is the Lipschitz constant of $\tilde{\Phi}$ from Assumption~\ref{ass:BINF}.
The function $W(x):=\tilde{W}(\Phi(x^p,w),w)$ also satisfies the upper bound~\eqref{eq:detect_grimm_1} using
\begin{align*}
{W}(x)=&c_1V_\eta(\Phi_2(x^p,w),w)+c_2\|\Phi_1(x^p,w)\|^2_{P_z}\\
\leq &c_1\tilde{c}_{o,u}\|\Phi_2(x^p,w)-\Phi_2(\pi_x(w),w)\|^2\\
&+c_2\lambda_{\max}(P_z)\|\Phi_1(x^p,w)-\Phi_1(\pi_x(w),w)\|^2\\
\leq &\max\{c_1\tilde{c}_{o,u},c_2\lambda_{\max}(P_z)\}L_\Phi^2\|x^p-\pi_x(w)\|^2,
\end{align*}
and thus satisfies Assumption~\ref{ass:detect}. 
\end{subequations}
\end{proof}
Due to the well-defined zero dynamics (Ass.~\ref{ass:zero}), minimizing $F_1$ in the look-ahead stage cost $\ell_d$ corresponds to an input regularization with respect to the input $u=\alpha(x)$. 
Based on this, the result in Proposition~\ref{prop:minphase_detect} can be intuitively interpreted in the form of a detectability notion.  
In particular, detectability ensures that for $(u,y)\equiv 0$, the plant state $x^p$ is asymptotically stable. 
The minimum phase property implies that the state $x^p-\pi_x(w)$ is asymptotically stable on the set $L_D$ (zero dynamics: $y\equiv 0$, $u=\alpha(x)$).
Hence, the minimum phase property is similar to detectability for a shifted input $\Delta u=u-\alpha(x)$ and replaces the i-IOSS condition (Ass.~\ref{ass:IOSS}) in the analysis.

We point out that this (implicit) input regularization w.r.t. $u=\alpha(x)$ is different than the input regularization w.r.t. $\pi_u(w)$ considered in Section~\ref{sec:reg}, since in general $\alpha(x^p,w)\neq \pi_u(w)$, except for $x^p=\pi_x(w)$.

The considered look-ahead stage cost satisfies $\ell_d(x_t,u_t)=y^2_t+y^2_{t+d+1}$ and hence it is possible to directly implement an MPC scheme with this look-ahead stage cost $\ell_d$, without explicitly using the BINF from Assumption~\ref{ass:BINF}. 
The same stage cost has also been suggested in~\cite[Equ.~(44)]{aguilar2014model} to study infinite horizon optimal regulation and approximations thereof.
Even though this stage cost can be implemented, in Theorem~\ref{thm:minPhase} we show that we obtain the same properties using a more standard output stage cost, albeit with a potentially  larger prediction horizon $N$.

\subsection{Theoretical analysis}
\label{sec:minPhase_3}
In the following, we show that, given the minimum phase property (Ass.~\ref{ass:minPhase}), the proposed output regulation MPC from Section~\ref{sec:reg} also ensures stability without any input regularization, i.e., by using the output stage cost $\ell_y(x_t)=h^2(x_t)=y^2_t$. 
The corresponding open-loop cost is denoted by $J^y_N(x_{\cdot|t},u_{\cdot|t}):=\sum_{k=0}^{N-1}\ell_y(x_{k|t})$ and the optimal value function by $V_N^y(x_t)$.
\begin{theorem}
\label{thm:minPhase}
Consider a SISO system and the MPC scheme~\eqref{eq:MPC}  with the stage cost $\ell(x,u)$ replaced by the output stage cost $\ell_y(x)$ ($R=0$). 
Let Assumptions~\ref{ass:regulator}--\ref{ass:increm_stab} and \ref{ass:BINF}--\ref{ass:minPhase} hold. 
Suppose further that $\pi_x$ from Ass.~\ref{ass:regulator} and $h$ are Lipschitz continuous. 
For any constant $\overline{Y}>0$, there exists a constant $N_{\overline{Y}}$, such that for $N>N_{\overline{Y}}$ and initial condition $x_0\in \mathbb{X}_{\overline{Y}}:=\{x\in\mathbb{X}|~2V^y_N(x)-\|\Phi_1(x)\|^2+W(x)\leq \overline{Y}\}$, the corresponding MPC problem is recursively feasible, the constraints are satisfied, and the regulator manifold $\mathcal{A}$ is exponentially stable for the resulting closed loop.
\end{theorem}
\begin{proof}
\begin{subequations}
The proof is structured as follows: We first show that the MPC formulation is equivalent to an MPC using the look-ahead stage cost $\ell_d$ and a semi-definite terminal cost.
Then, we exploit the fact that $\ell_d$ satisfies the detectability condition (Ass.~\ref{ass:detect}) and extend the proof of Theorem~\ref{thm:MPC}/Corollary~\ref{corol:main}.\\
\textbf{Part I. }The output stage cost $\ell_y(x)$ and the look-ahead stage cost $\ell_d(x,u)$ are such that for any trajectory $x_t,u_t$ satisfying the dynamics, we have $\ell_d(x_t,u_t)=\ell_y(x_t)+\ell_y(x_{t+d+1})$, compare the BINF~\eqref{eq:BINF}. 
Define the open-loop cost $J^d_N(x_{\cdot|t},u_{\cdot|t}):=\sum_{k=0}^{N-1}\ell_d(x_{k|t},u_{k|t})$.
For any $N\geq d+1$, we have 
\begin{align}
\label{eq:open_loop_d}
2J^y_N(x_{\cdot|t},u_{\cdot|t})=&\|\Phi_1(x_t)\|^2+J^d_{N-d-1}(x_{\cdot|t},u_{\cdot|t})\nonumber\\
&+\|\Phi_1(x_{N-d-1|t})\|^2,
\end{align}
where we used the fact that $\|\Phi_1(x_t)\|^2=\sum_{k=0}^{d}\ell_y(x_{t+k})$, compare the BINF~\eqref{eq:BINF}. 
Hence, minimizing the cost $J_N^y$ yields the same optimal input as minimizing the look-ahead stage cost $\ell_d$ over a shorter prediction horizon and adding a positive-semi definite terminal cost.\\
\textbf{Part II. }
In the following, we analyse the closed loop using the shifted value function 
\begin{align}
\label{eq:open_loop_opt_d}
\tilde{V}_N(x_t):=&2V^{y}_N(x_t)-\|\Phi_1(x_t)\|^2\nonumber\\
\stackrel{\eqref{eq:open_loop_d}}{=}&J_{N-d-1}^d(x_{\cdot|t}^*,u_{\cdot|t}^*)+\|\Phi_1(x_{N-d-1|t}^*)\|^2.
\end{align}
Similar to Proposition~\ref{prop:stab}, the function $\tilde{V}_N$ also satisfies Assumption~\ref{ass:stab} with the same constant $\delta_s>0$ and $\tilde{\gamma}_s=2\gamma_s$.
Furthermore, due to Proposition~\ref{prop:minphase_detect}, we have $W(f(x,u))-W(x)\leq- \epsilon_o\sigma(x)+\ell_d(x,u)$. 
Consider the Lyapunov candidate function $Y_N^y(x):=W(x)+\tilde{V}_N(x)$, which also satisfies the upper and lower bound in~\eqref{eq:Lyap_1} from Theorem~\ref{thm:MPC} with $\tilde{\gamma}_{\overline{Y}}:=\max\{\tilde{\gamma}_s+\gamma_o,\overline{Y}/\delta_s\}$ and $\mathbb{X}_{\overline{Y}}=\{x\in\mathbb{X}|~Y^y_N(x)\leq \overline{Y}\}$.
Analogous to \eqref{eq:rotated}--\eqref{eq:turnpike} in the proof of Theorem~\ref{thm:MPC}, we use $W,\ell^y$ non-negative to ensure
\begin{align*}
\epsilon_o\sum_{k=0}^{N-d-2}\sigma(x_{k|t}^*)\stackrel{\eqref{eq:detect_grimm_2}}{\leq} W(x_t)+J_{N-d-1}^d(x_{\cdot|t}^*,u_{\cdot|t}^*)
\stackrel{\eqref{eq:open_loop_opt_d}}{\leq} Y_N^y(x_t). 
\end{align*}
Hence, there exists a $k_x\in\{1,\dots,N-d-2\}$, such that
\begin{align}
\label{eq:turnpike_term}
\sigma(x_{k_x|t}^*)\leq\frac{Y_N^y(x_t)}{\epsilon_o(N-d-2)}\leq\frac{ \min\{\overline{Y},\tilde{\gamma}_{\overline{Y}}\sigma(x_t)\}}{\epsilon_o(N-d-2)}.
\end{align}
Given $N\geq N_0:=2+d+\frac{\overline{Y}}{\delta_s\epsilon_o}$, this implies $x_{k_x|t}^*\in\mathbb{X}_\delta$. 
Denote $\ell^d_k:=\ell_d(x^*_{k|t},u^*_{k|t})$. 
Similar to~\eqref{eq:Value_dec}, we obtain
\begin{align*}
&\tilde{V}_N(x_{t+1})+\ell^d_0\leq \sum_{k=0}^{k_x-1}\ell_k^d+\tilde{V}_{N-k_x+1}(x^*_{k_x|t})\\
\stackrel{\eqref{eq:stab_grimm},\eqref{eq:turnpike_term}}{\leq}&\tilde{V}_N(x_t)+\dfrac{\tilde{\gamma}_s\tilde{\gamma}_{\overline{Y}}}{\epsilon_o(N-d-2)}\sigma(x_t). 
\end{align*}
The remainder of the proof is analogous to Theorem~\ref{thm:MPC} and Corollary~\ref{corol:main} resulting in
$\alpha_N:=1-\dfrac{\tilde{\gamma}_s\tilde{\gamma}_{\overline{Y}}}{\epsilon_o^2(N-d-2)}>0$ for $N>N_{\overline{Y}}:=\max\{N_0,N_1\}$, $N_1:=\tilde{\gamma}_s\tilde{\gamma}_{\overline{Y}}/\epsilon_o^2+d+2$.
\end{subequations}
\end{proof}
For the special case of minimum phase systems, this result shows that we can use the MPC formulation~\eqref{eq:MPC} without any input regularization to solve the output regulation problem. 
We point out that the zero dynamics are also vital in the classical output regulation literature~\cite{isidori2013zero} and while there exist results for non-minimum phase systems, ``most methods [$\dots$] only address systems in normal form with a (globally) stable zero dynamics''~\cite{priscoli2009dissipativity}.

We emphasize that in order to apply the proposed MPC scheme, we do not need to solve the regulator equations~\eqref{eq:reg}. 
This is only possible, since we do not use a positive definite stage cost $\ell$ or terminal ingredients, both of which would drastically simplify the theoretical analysis but would necessitate knowledge of $\pi_x(w)$.
Thus, compared to Proposition~\ref{prop:output_regulation} (classical solution), the proposed MPC scheme has the following advantages:
\begin{itemize}
\item Explicit solution to the regulator equations~\eqref{eq:reg} is not required,
\item No explicit stabilizing controller $\kappa$ (Ass.~\ref{ass:increm_stab}) is needed,
\item The MPC scheme enjoys a larger region of attraction. 
\end{itemize}
Compared to the MPC schemes in~\cite{falugi2013tracking,kohler2020nonlinear} and \cite{limon2018nonlinear,magni2005solution}, we do not pose any periodicity conditions on $w$ or restrict ourselves to constant values $w$.
The restriction to minimum phase systems will be relaxed in Section~\ref{sec:input} using a modified MPC formulation.

\subsection{Discussion}
\label{sec:minPhase_4}
\begin{remark}
\label{rk:MIMO_min_phase}
(MIMO systems)
The results in this section can be naturally extended to (square) MIMO systems with the output stage cost $\ell_y(x):=\|y\|_Q^2$ with $Q=\emph{diag}(q_i)$. 
In this case, the Byrnes-Isidori normal form (Ass.~\ref{ass:BINF}) contains integrator states $z^{i,k}$ and nonlinear maps $F_1^i$ for each output component $y^i$, with  different relative degrees $d_i$.  
Assumptions~\ref{ass:zero}--\ref{ass:minPhase} remain unchanged with $F_1=(F_1^1,\dots ,F_1^p)$. 
Proposition~\ref{prop:minphase_detect} remains true with the look-ahead cost $\ell_d(x_t,u_t)=\ell_y(x_t)+\sum_{i=1}^p q_i (y^i_{t+d_i+1})^2$. 
In Theorem~\ref{thm:minPhase} we consider $J_{N-\overline{d}-1}$ with $\overline{d}=\max_i d_i$ and obtain the different  non-negative ``terminal cost''  $\|\Phi_1(x_{N-\overline{d}-1|t})\|_{Q_d}^2+\sum_{i=1}^p\sum_{k=N-(\overline{d}-d_i)}^{N-1}q_i(y^i_{k|t})^2$ with $Q_d=\text{diag}(q_i)\in\mathbb{R}^{\sum_{i=1}^p (1+d_i)\times \sum_{i=1}^p (1+d_i)}$. The remainder of the proof remains unchanged.
\end{remark}

\begin{remark}
\label{rk:term_cost}
(Implicit terminal cost and extremely short prediction horizons)
The analysis contains a terminal cost $\|\Phi_1(x)\|^2=J^y_{d+1}(x_{N|t})$, which is locally equivalent to the value function $V^y_N$, $N\geq d+1$, with the optimal input $u=\alpha(x)$ due to Assumptions~\ref{ass:BINF}--\ref{ass:zero}. 
Thus, in the absence of constraints, a horizon $N>\overline{d}+1$ is sufficient to ensure stability for such minimum phase systems, which can be significantly less conservative than the usual bounds obtained in MPC without terminal constraints~\cite{grune2012nmpc,reble2012unconstrained,tuna2006shorter}.
We conjecture that stronger guarantees regarding  the prediction horizon $N_{\overline{Y}}$ and the suboptimality index $\alpha_N$ can be derived, even in the presence of state and input constraints.
\end{remark}

\begin{remark}
\label{rk:example_not_detect}
(Ass.~\ref{ass:detect} does not hold with $\ell_y$)
Although Theorem~\ref{thm:minPhase} ensures stability and utilizes a proof similar to Thm.~\ref{thm:MPC}, this was only possible by utilizing the look-ahead stage cost $\ell_d$ in the analysis.  Assumption~\ref{ass:detect} is in general not valid with the output stage cost $\ell_y$. 
Consider the trivial FIR filter $y_{t}=u_{t-2}$, with $x_t=(u_{t-1},u_{t-2})$, $\sigma(x)=\|x\|^2$, which clearly satisfies the conditions in Thm.~\ref{thm:minPhase}.
Considering Inequality~\eqref{eq:detect_grimm_2} for $x=(0,0)$ implies that $W((u,0))=0$ for all $u\in\mathbb{R}$. 
Now consider $x=(x_0,0)$ with $x_0\neq 0$ and $u\in\mathbb{R}$: Inequality~\eqref{eq:detect_grimm_2} implies $W((u,x_0))\leq W((x_0,0))-\epsilon_ox_0^2+0<0$, which contradicts the assumption that $W$ is non-negative. 
Thus, this system does not satisfy Assumption~\ref{ass:detect} with the output stage cost $\ell_y$.
\end{remark}

\section{Incremental input regularization}
\label{sec:input}
The MPC design using an input regularization (Sec.~\ref{sec:reg}) can only be implemented if the optimal feedforward input $\pi_u(w)$ is known, while the analysis in Section~\ref{sec:minPhase} is  restricted to minimum phase systems.
In this section we show how these restrictions can be relaxed by using an  incremental input regularization. 
The proposed formulation is presented in Section~\ref{sec:input_1}. 
The theoretical analysis is contained in Section~\ref{sec:input_2} and a discussion can be found in Section~\ref{sec:input_3}. 

\subsection{Incremental input formulation for periodic signals}
\label{sec:input_1}
The main idea is to reformulate the problem, such that the optimal feedforward input vanishes by using an incremental input regularization in the MPC formulation.
To allow for this reformulation, we focus on periodic exogenous signals.
\begin{assumption}
\label{ass:exo_periodic}
(Periodic exogenous signals)
There exists a known period length $T\in\mathbb{N}$, such that $w_{t+T}=w_t$ for all $t\geq 0$ with $w$ evolving according to~\eqref{eq:sys_exo}. 
\end{assumption}
In the classical literature on output regulation, compare e.g.~\cite{isidori1990output,castillo1993nonlinear,pavlov2006uniform,byrnes2012output}, the exosystem is assumed to be neutrally/poisson stable, which in the linear case reduces to either constant or harmonic/periodic exogenous signals $w$ and hence Assumption~\ref{ass:exo_periodic} holds with $T$ being the least common multiple of the different period lengths.  
We point out that the complexity of the following MPC formulation does not scale with the period length $T$, but depends on the prediction horizon $N$, and hence large values of $T$ are not a problem. 

Define a memory state for the past applied inputs as $\xi_t:=(u_{t-1},\dots u_{t-T})\in\mathbb{U}^T$. 
The proposed input regularized MPC formulation is based on the following optimization problem:
\begin{subequations}
\label{eq:MPC_input}
\begin{align}
\label{eq:MPC_input_1}
V^a_N(x^p_t,w_t,\xi_t):=&\min_{u_{\cdot|t}\in\mathbb{U} ^N}\sum_{k=0}^{N-1}\|y_{k|t}\|_Q^2+\|\Delta u_{k|t}\|_{R}^2\\
\label{eq:MPC_input_2a}
&x^p_{0|t}=x^p_t,\quad w_{0|t}=w_t,\\
\label{eq:MPC_input_2b}
& u_{-j|t}=u_{t-j},\\
\label{eq:MPC_input_3a}
&x^p_{k+1|t}=f^p(x_{k|t},u_{k|t},w_t),\\
\label{eq:MPC_input_3b}
&w_{k+1|t}=s(w_{k|t}),\\
\label{eq:MPC_input_3c}
& u_{k|t}=u_{k-T|t}+\Delta u_{k|t},\\
\label{eq:MPC_input_3d}
&y_{k|t}=h(x^p_{k|t},u_{k|t},w_{k|t}),\\
\label{eq:MPC_input_4}
&(x^p_{k|t},u_{k|t})\in\mathcal{Z}^p,\\
&k=0,\dots,N-1,\nonumber\\
&j=\max\{T-N,0\}+1,\dots,T,\nonumber
\end{align}
\end{subequations}
with positive definite matrices $Q$, $R$.  
The difference to the MPC formulation in Sections~\ref{sec:reg} and \ref{sec:minPhase} is the usage of an incremental input regularization $\|\Delta u\|_R^2$ that penalizes nonperiodic input signals $u$. 
Although the optimal feedforward solution $(\pi_x(w),\pi_u(w))$ is unknown, we know that $w$ and hence $\pi_u(w)$ is $T$-periodic (Ass.~\ref{ass:exo_periodic}). 
Thus, intuitively speaking, we know that the optimal solution should drive $(y,\Delta u)$ to the origin using the considered stage cost with incremental input regularization. 
We note that except for the case of constant exogenous signals (T=1, Sec.~\ref{sec:input_3}), the considered input cost penalizing non-periodicity differs structurally from existing trajectory tracking MPC formulations~\cite{koehler2020quasi,faulwasser2015nonlinear,falugi2013tracking,kohler2020nonlinear}, which require a known input (and state) reference. 
A similar penalty on nonperiodic trajectories was used in~\cite{gutekunst2020economic} for periodic optimal control.

\subsection{Theoretical analysis}
\label{sec:input_2}
In the following, we show that the input regularized optimization problem~\eqref{eq:MPC_input} is equivalent to Problem~\eqref{eq:MPC} for an augmented plant model with modified state and  input.  

\subsubsection*{Augmented plant}
Consider the augmented plant model
\begin{subequations}
\label{eq:periodic_model}
\begin{align}
\label{eq:periodic_model_1}
{x}^{ap}_t:=&(x^p_t,\xi_t),~{u}^{a}_t:=u_t-u_{t-T},
\end{align}
Define the block cyclic permutation matrix $E_0\in\mathbb{R}^{mT\times mT}$ and the selection matrices $E_1$, $E_2\in\mathbb{R}^{mT\times m}$:
{\small
\begin{align*}
E_0:=&\begin{pmatrix}
0_m&\hdots&0_m&I_m\\
I_m&\hdots &0_m&0_m\\
0_m&\ddots&\vdots&\vdots\\
\vdots&\ddots&0_m&0_m\\
0_m&0_m& I_m&0_m
\end{pmatrix},
E_1:=\begin{pmatrix}
I_m\\0_m\\\vdots\\0_m
\end{pmatrix},
E_2:=\begin{pmatrix}
0_m\\\vdots\\0_m\\I_m
\end{pmatrix}.\nonumber
\end{align*}}
This matrix satisfies $\Pi_{k=0}^{T-1}E_0=I_{mT}$ and the eigenvalues of  $E_0$ are $\lambda_k=e^{i 2\pi k/T}$, $k=0,\dots,T-1$, all with a geometric and algebraic multiplicity of $m$ (due to the block structure).
We note that $E_1^\top E_0=E_2^\top$. 
The dynamics of the memory state $\xi$ and the input $u$ can be compactly expressed as 
\begin{align} 
\label{eq:periodic_model_4}
&\xi_{t+1}=E_0 \xi_t+E_1{u}^a_t,~u_t=E_2^\top \xi_t+{u}^a_t.
\end{align}
Thus, the augmented plant dynamics ${f}^{ap}$ can be expressed as
\begin{align}
\label{eq:periodic_model_5}
{f}^{ap}({x}^{ap}_t,{u}^a_t,w_t):=&(f^p(x^p_t,E_2^\top\xi_t+{u}^a_t,w_t),E_0\xi_t+E_1{u}^a_t).
\end{align}
The corresponding constraint sets are given by
\begin{align}
\label{eq:periodic_model_6}
\mathbb{X}^{ap}:=&\{{x}^{ap}=(x^p,\xi)|~x\in\mathbb{X},~\xi\in\mathbb{U}^T\},~\mathbb{U}^a=\mathbb{R}^m,\\ 
\mathcal{Z}^{ap}:=&\{(x^p, \xi,{u}^a)\in\mathbb{X}^{ap}\times\mathbb{U}^a|~({x}^p,{u}^a+E_2^\top\xi)\in\mathcal{Z}^p\}. \nonumber
\end{align}
The overall augmented state ${x}^a=({x}^{ap},w)$ is subject to the exosystem dynamics~\eqref{eq:sys_exo} and the following output equation
\begin{align}
\label{eq:periodic_model_7}
y_t=h^{a}(x^{ap}_t,u^a_t,w_t):=&h(x^{p}_t,E_2^\top \xi_t+{u}^a_t,w_t).
\end{align}
\end{subequations}
\begin{lemma}
\label{lemma:augmented_MPC}
The optimization problem~\eqref{eq:MPC} with the augmented state ${x}^a=({x}^{ap},w)$, input ${u}^a$, the dynamics~\eqref{eq:periodic_model},
and stage cost $\ell^a({x}^a,{u}^a)=\|h^a(x^{ap},u^a,w)\|_{Q}^2+\|u^a\|_R^2$ using incremental input regularization is equivalent to the optimization problem~\eqref{eq:MPC_input}. 
\end{lemma}
Lemma~\ref{lemma:augmented_MPC} ensures that we can analyse the output regulation MPC with incremental input regularization~\eqref{eq:MPC_input} using the results in Theorem~\ref{thm:MPC} and Corollary~\ref{corol:main}. 
We only need to show that the augmented system also satisfies  Assumptions~\ref{ass:regulator}, \ref{ass:increm_stab}, \ref{ass:IOSS}, which will be done in the following.
\begin{proposition}
\label{prop:periodic_cond}
Let Assumptions~\ref{ass:regulator}, \ref{ass:increm_stab} and \ref{ass:exo_periodic} hold.
The agumented system~\eqref{eq:periodic_model} also satisfies Assumptions~\ref{ass:regulator}, \ref{ass:increm_stab}.
\end{proposition}
\begin{proof}
\textbf{Assumption~\ref{ass:regulator}}: 
Define the recursive composition $s^{k+1}:=s\circ s^{k}$ with the dynamics $s$ of the exosystem in~\eqref{eq:sys_exo}. 
Given that the original plant satisfies the regulator equations (Ass.~\ref{ass:regulator}) and the exogenous signal is $T$-periodic, the augmented system also satisfies the regulator equations~\eqref{eq:reg} with ${\pi}^a_x(w)=(\pi_x(w),\pi_u(s^{T-1}(w)),\dots,\pi_u(w))$, ${\pi}^a_u(w)=0$. Similarly, we have $({\pi}^a_x(w),{\pi}^a_u(w))\in\text{int}(\mathcal{Z}^{ap})$ for all $w\in\mathbb{W}$.\\
\textbf{Assumption~\ref{ass:increm_stab}:}  
Consider $(x^{ap},u^a,z^{ap},v^a,w)\in\mathcal{Z}^{ap}\times\mathcal{Z}^{ap}\times\mathbb{W}$  with ${z}^{ap}=(z^p,\Xi)$, $x^{ap}=(x^p,\xi)$, $\xi,\Xi\in\mathbb{U}^T$. 
We obtain the  input $u$ and plant dynamics $x^p$ from Assumption~\ref{ass:increm_stab} with
\begin{align*}
{u}^a={\kappa}^a({x}^{ap},z^{ap},w,v^a):= \kappa(x^p,z^p,w,v)-E_2^\top\xi.
\end{align*} 
Denote that successor state of $\xi$ by
$\xi^+=E_0\xi+E_1u_a=\tilde{E}_0\xi+E_1u$,  $\tilde{E}_0:=E_0-E_1E_2^\top$. 
The matrix $\tilde{E}_0$ corresponds to the dynamics of a finite impulse response (FIR) system and is thus nilpotent. 
Hence, the $\xi$ dynamics are incrementally input to state stable (i-ISS) w.r.t $u$ with an arbitrarily small contraction rate and thus w.l.o.g. with $\rho_s\in(0,1)$ from Assumption~\ref{ass:increm_stab}, i.e., there exists a positive definite matrix $P_{\xi,s}=P_{\xi,s}^\top\succ 0$, such that
\begin{align}
\label{eq:increm_ISS}
\|\xi^+-\Xi^+\|_{P_{\xi,s}}^2\leq \rho_s\|\xi-\Xi\|_{P_{\xi,s}}^2+\|{u}-{v}\|^2.
\end{align}
Given $\kappa$ Lipschitz continuous with $\kappa(z^p,z^p,w,v)=v$ and~\eqref{eq:increm_b}, we have $\|u-v\|^2\leq L_\kappa V_s(x^p,z^p,w)$ with some $L_\kappa> 0$.
For any ${\rho}_{a,s}\in(\rho_s,1)$, choose $c=\frac{{\rho}_{a,s}-\rho_s}{L_\kappa}>0$. 
The joint incremental Lyapunov function ${V}_{a,s}({x}^{ap},z^{ap},w):=V_s(x^p,z^p,w)+c\cdot \|\xi-\Xi\|_{P_{\xi,s}}^2$ satisfies~\eqref{eq:increm_a} with
\begin{align*}
&V_{a,s}({f}^{ap}({x}^{ap},{u}^a,w),{f}^{ap}({z}^{ap},{v}^a,w),s(w))\\
\stackrel{\eqref{eq:increm_a},\eqref{eq:increm_ISS}}{\leq}& (\rho_s+cL_\kappa) V_s(x^p,z^p,w)+\rho_s \|\xi-\Xi\|_{P_{\xi,s}}^2\\
\leq& {\rho}_{a,s} {V}_{a,s}({x}^{ap},z^{ap},w).
\end{align*}
Conditions~\eqref{eq:increm_b} holds with ${c}_{a,s,u}=\max\{c_{s,u},c\lambda_{\max}(P_{\xi,s})\}$ and ${c}_{a,s,l}=\min\{c_{s,l},c\lambda_{\min}(P_{\xi,s})\}$.
\end{proof}
The main benefit of this formulation is that we do not need $\pi_u(w)$ for the implementation, since the incremental input formulation in combination with the assumed periodicity guarantees $\pi^a_u(w)=0$.

\subsubsection*{Detectability and the nonresonance condition}
The augmented plant ${x}^{ap}$ is a series connection of two detectable systems: $x^p$ and $\xi$ with $u$ being the respective input and output. 
To ensure that this system also satisfies the detectability condition (Ass.~\ref{ass:IOSS}), we need an additional \textit{nonresonance condition}. 
\begin{assumption}
\label{ass:nonres_nonlin}
(Nonlinear nonresonance condition)
There exist a  continuous incremental storage function $V_R:\mathbb{X}^{ap}\times\mathbb{X}^{ap}\times\mathbb{W}\rightarrow\mathbb{R}_{\geq 0}$ and constants $c_{R,u}$, $c_{R}>0$, such that for any $(z^{ap},v^a,x^{ap},u^a,w)\in\mathcal{Z}^{ap}\times\mathcal{Z}^{ap}\times\mathbb{W}$, we have
\begin{subequations}
\label{eq:nonlin_nonresonance}
\begin{align}
\label{eq:nonlin_nonresonance_1}
&V_R({x}^{ap},z^{ap},w)\leq c_{R,u}\|{x}^{ap}-{z}^{ap}\|^2,\\
\label{eq:nonlin_nonresonance_2}
&V_{R}({f}^{ap}({x}^{ap},{u}^a,w),{f}^{ap}({z}^{ap},{v}^a,w),s(w))\\
\leq&V_R({x}^{ap},z^{ap},w)-\|u-v\|^2\nonumber\\
&+ c_{R}(\|{u}^a-{v}^a\|^2+\|h^a(x^{ap},u^a,w)-h^a(z^p,v^a,w)\|^2),\nonumber
\end{align}
\end{subequations}
with ${x}^{ap}=(x^p,\xi)$, ${z}^{ap}=(z^p,\Xi)$, $u=E_2^\top \xi+{u}^a$, and $v=E_2^\top \Xi+{v}^a$.
\end{assumption}
Given that condition~\eqref{eq:nonlin_nonresonance} corresponds to an incremental dissipativity condition, it can be verified using the results in~\cite{verhoek2020convex} based on differential dissipativity.
Loosely speaking, conditions~\eqref{eq:nonlin_nonresonance} imply that if two systems have a similar initial condition, produce a similar output and are driven by a similar incremental input ${u}^a/{v}^a$, then the  input $u/v$ applied to the plant has to be similar. 
In particular, if both systems are driven by a periodic input $u,v$ and generate the same output trajectory $y$, then the two periodic input trajectories $u,v$ must be equivalent. 
Thus, this condition excludes the possibility of two distinct periodic inputs $u,v$ resulting in the same output $y$. 
This condition seems to be a relaxed version of \textit{input detectability/observability}, as for ${u}^a={v}^a=0$ (periodic inputs) it essentially requires that $y\equiv 0$ implies  $u\equiv 0$ (assuming zero initial conditions), similar to~\cite[Def.~3]{hou1998input}. 
In the linear case, this is equivalent to assuming that the poles generating here a $T$-periodic input signal $u$ with~\eqref{eq:periodic_model_4} (assuming ${u}^a=0$) are not cancelled by zeros of the plant, which corresponds to the well established \textit{nonresonance condition}, compare Section~\ref{sec:linear} for a detailed proof.
We point out that in~\cite{marconi2004non} a different nonlinear extension of the \textit{nonresonance condition}  has been proposed, which is characterized using a rank condition on the lie derivatives as opposed to a dissipativity characterization. 
Although both characterizations correspond to the classical \textit{nonresonance condition} in the linear case, the considered formulation using dissipation inequalities with a storage function allows us to directly use an i-IOSS Lyapunov function to establish detectability of the augmented plant, as shown in the following proposition.  
\begin{proposition}
\label{prop:non_res}
Let Assumptions~\ref{ass:IOSS}, \ref{ass:exo_periodic} and \ref{ass:nonres_nonlin} hold.
The augmented system~\eqref{eq:periodic_model} satisfies Ass.~\ref{ass:IOSS}. 
\end{proposition}
\begin{proof}
\begin{subequations}
\label{eq:proof_non_res}
Consider $(x^{ap},u^a,z^{ap},v^a,w)\in\mathcal{Z}^{ap}\times\mathcal{Z}^{ap}\times\mathbb{W}$  with ${z}^{ap}=(z^p,\Xi)$, $x^{ap}=(x^p,\xi)$, $\xi,\Xi\in\mathbb{U}^T$. 
First note that the linear dynamics of $\xi$ with the input ${u}^a$ and the output $u=E_2^\top \xi+{u}^a$ are observable. 
Thus, there exists a quadratic  i-IOSS Lyapunov function $V_\xi(\xi,\Xi)=\|\xi-\Xi\|_{P_{\xi,o}}^2$ (cf.~\cite{cai2008input}) satisfying 
\begin{align}
\label{eq:IOSS_xi}
V_\xi(\xi^+,\Xi^+)\leq \rho_\xi V_\xi(\xi,\Xi)+\|{u}^a-{v}^a\|^2+\|u-v\|^2,
\end{align}
with $\rho_\xi\in(0,1)$, $\xi^+=E_0\xi+E_1{u}^a$, $u=E_2^\top \xi+{u}^a$, $\Xi^+=E_0\Xi+E_1{v}^a$, $v=E_2^\top \Xi+{v}^a$. 
Consider the candidate i-IOSS Lyapunov function
\begin{align}
\label{eq:IOSS_joint}
&{V}_{a,o}({x}^{ap},z^{ap},w)\\
:=&V_o(x^p,z^p,w)+V_\xi(\xi,\Xi)+c_2 V_R({x}^{ap},z^{ap},w),\nonumber
\end{align}
with  $c_2=c_{o,1}+1$. 
The lower and upper bounds~\eqref{eq:IOSS_1} follow directly with with ${c}_{a,o,u}:=\max\{c_{o,u},\lambda_{\max}(P_{\xi,0})\}+c_2 c_{R,u}$ and ${c}_{a,o,l}:=\min\{c_{o,l},\lambda_{\min}(P_{\xi,o})\}$. 
The i-IOSS condition~\eqref{eq:IOSS_2} holds with 
\begin{align*}
&{V}_{a,o}({f}^{ap}({x}^{ap},{u}^a,w),{f}^{ap}({z}^{ap},{v}^a,w),s(w))\\
&-{V}_{a,o}({x}^{ap},z^{ap},w)\\
\stackrel{\eqref{eq:IOSS_2},\eqref{eq:nonlin_nonresonance_2},\eqref{eq:IOSS_xi}}{\leq} &-(1-\rho_o) V_o(x^p,z^p,w)+(c_{o,1}+1-c_2)\|u-v\|^2\\
&+(c_{o,2}+c_2c_{R})\|h(x^p,u,w)-h(z^p,v,w)\|^2\\
&-(1-\rho_\xi)V_\xi(\xi,\Xi)+(1+c_2c_{R})\|{u}^a-{v}^a\|^2\\
\stackrel{\eqref{eq:IOSS_1}}{\leq}
&-{\epsilon}_a(\|x^p-z^p\|^2+\|\xi-\Xi\|^2)+{c}_{a,o,1}\|{u}^a-{v}^a\|^2\\
&+{c}_{a,o,2}\|h(x^p,u,w)-h(z^p,v,w)\|^2\\
\leq &-(1-\rho_{a,o}) {V}_{a,o}({x}^{ap},z^{ap},w)+{c}_{a,o,1}\|{u}^a-{v}^a\|^2\\
&+{c}_{a,o,2}\|h^{ap}(x^{ap},u^a,w)-h^{ap}(z^{ap},v^a,w)\|^2,
\end{align*}
with
$\epsilon_a:=\min\{(1-\rho_o)c_{o,l},(1-\rho_\xi)\lambda_{\min}(P_{\xi,o})\}>0$, 
 $\rho_{a,o}:=1-\frac{\epsilon_a}{{c}_{a,o,u}}\in(0,1)$, ${c}_{a,o,1}=1+c_2c_{R}$, ${c}_{a,o,2}=c_{o,2}+c_2c_{R}$.
\end{subequations}
\end{proof}
\subsubsection*{Final result}
With Propositions~\ref{prop:periodic_cond}--\ref{prop:non_res} and Lemma~\ref{lemma:augmented_MPC}, we can  summarize the theoretical properties of the incremental input regularized MPC scheme~\eqref{eq:MPC_input}.
\begin{corollary}
\label{corol:input}
Suppose the plant~\eqref{eq:sys} satisfies Assumptions~\ref{ass:regulator},  \ref{ass:increm_stab}, \ref{ass:IOSS}, and Assumptions~\ref{ass:exo_periodic}--\ref{ass:nonres_nonlin} hold. 
Suppose further that $\pi_x$, $\pi_u$, $s$ and $h$ are Lipschitz continuous. 
For any constant $\overline{Y}>0$,  there exist a constant $N_{\overline{Y}}$, such that for $N>N_{\overline{Y}}$ and  initial condition $(x_0,\xi_0)={x}^a_0\in \mathbb{X}^a_{\overline{Y}}:=\{{x}^a\in\mathbb{X}^a|~V_N^a(x^a)+W(x^a)\leq \overline{Y}\}$, the MPC problem~\eqref{eq:MPC_input} is recursively feasible, the constraints are satisfied  and the (augmented) regulator manifold $\mathcal{A}^{a}:=\{{x}^a|~{x}^{ap}={\pi}^a_x(w)\}$ is exponentially stable for the resulting closed loop.
\end{corollary}
\begin{proof}
 Lemma~\ref{lemma:augmented_MPC} ensures that the MPC problem~\eqref{eq:MPC_input} corresponds to the MPC problem~\eqref{eq:MPC} for an augmented plant ${x}^{ap}$ and Propositions~\ref{prop:periodic_cond}--\ref{prop:non_res} ensure that this augmented plant satisfies Assumptions~\ref{ass:regulator}, \ref{ass:increm_stab}, \ref{ass:IOSS} with $\pi_x^a$ Lipschitz continuous.
 Thus, the closed-loop properties follow from Corollary~\ref{corol:main}.
\end{proof}

 \subsection{Discussion}
\label{sec:input_3}
The computational demand of the proposed approach scales with the prediction horizon $N$, but not directly with the period length $T$. 
However, the sufficient prediction horizon $N_{\overline{Y}}$ may increase compared to bounds derived for the MPC scheme in Section~\ref{sec:reg}.

\subsubsection*{Existing MPC solutions for periodic problems}
For the special case of periodic signals $w$, there also exist competing approaches to solve the regulator problem. Given $w_0$ and the period length $T$, the $T$-periodic trajectory $(\pi_x(w_t),\pi_u(w_t))$ can be obtained by solving one (potentially large) nonlinear program (NLP), as suggested in~\cite{falugi2013tracking}. 
Then the output regulation problem reduces to the problem of stabilizing a given state and input trajectory, for which established MPC approaches with and without terminal ingredients exist, compare~\cite{koehler2020quasi,faulwasser2015nonlinear} and \cite{kohler2018nonlinear}, respectively. 
If we consider online changing operating conditions or the error feedback setting (Remark~\ref{rk:error_feedback}), the estimates for $w_t$ may change online and thus the large scale NLP would have to be repeatedly solved during online operation.
The problem of online recomputing a periodic reference trajectory can be integrated in the MPC formulation using artificial reference trajectories, as e.g. done in~\cite{kohler2020nonlinear}. 
The additional complexity of recomputing a periodic solution can be further reduced using a partially decoupled MPC design~\cite[Sec.~3.4]{kohler2020nonlinear}. 

\subsubsection*{Offset-free setpoint tracking - incremental input penalty} 
The problem of offset-free setpoint tracking is a special case with $s(w)=w$ and $T=1$.
In this case, $\Delta u$ penalizes the change in the control input $u$, which is quite common in the MPC literature, especially in case of offset-free setpoint tracking, compare e.g.~\cite{betti2013robust,magni2005solution,muske2002disturbance} and~\cite[Cor.~4]{magni2001output}. 
Thus, the proposed formulation is rather intuitive and similar to existing standard approaches for tracking MPC. 
For comparison, in~\cite{magni2005solution} a linear dynamic controller is used to characterize the terminal cost and set and in~\cite{limon2018nonlinear} artificial setpoints are used to track changing setpoints. 
The issue of estimating the disturbances has, e.g., been treated in~\cite{muske2002disturbance,morari2012nonlinear} for linear and nonlinear systems with disturbance observers and can also be treated in the proposed framework, compare Remark~\ref{rk:error_feedback}. 
\subsubsection*{Nonresonance condition = tracking condition}
In the case of nonlinear setpoint tracking MPC, it is often assumed that there exists a unique (Lipschitz continuous) map from any output $y$ to a corresponding steady state and input $(x^p_s,u_s)$, compare e.g. \cite[Ass.~1]{limon2018nonlinear} or \cite[Ass.~1]{magni2005solution}. 
Given $f^p,h$ continuously differentiable, this condition is equivalent to a rank condition on the linearized system (cf.~\cite[Remark~1]{limon2018nonlinear}, \cite[Lemma~1.8]{rawlings2017model}), which is equivalent to the nonresonance condition for constant exogenous signals, compare~\cite{marconi2004non}. 
We point out that the rank-based and dissipation-based nonresonance characterizations are equivalent in the linear case (cf.~Prop.~\ref{prop:linear_equiv} below).
Thus, intuitively this tracking condition~\cite[Ass.~1]{limon2018nonlinear} is strongly related (if not equivalent) to the dissipation-based characterization in Assumption~\ref{ass:nonres_nonlin} for $T=1$. 
Furthermore, a similar characterization to~\cite[Ass.~1]{limon2018nonlinear} can be used for periodic trajectories~\cite[Ass.~6]{kohler2020nonlinear}, which seems to be an alternative characterization for the property in Assumption~\ref{ass:nonres_nonlin}.

\section{Special case - linear systems}
\label{sec:linear}
In this section, we consider the special case of linear systems
\begin{subequations}
\label{eq:sys_lin}
\begin{align}
\label{eq:sys_dyn_lin}
x^p_{t+1}&=Ax^p_t+Bu_t+P_x w_t,\\
\label{eq:sys_exo_lin}
w_{t+1}&=Sw_t,\\
\label{eq:sys_output_lim}
y_t&=Cx^p_t+Du_t-P_y w_t,
\end{align}
\end{subequations}
and discuss how the Assumptions in Sections~\ref{sec:reg}--\ref{sec:input} simplify. 
In addition, Proposition~\ref{prop:linear_equiv} shows that in the linear case the dissipation characterization in Assumption~\ref{ass:nonres_nonlin} is equivalent to the classical rank based nonresonance condition.
\subsection{Stabilizability/Detectability}
Assumption~\ref{ass:increm_stab} reduces to stabilizability of $(A,B)$ and Assumption~\ref{ass:IOSS} reduces to detectability of $(A,C)$ (cf.~\cite{cai2008input}).
The regulator equations~\eqref{eq:reg} (Ass.~\ref{ass:regulator}) reduce to
\begin{align}
\label{eq:reg_lin}
\Pi S=A\Pi+B\Gamma+P_x,~
0=C\Pi+D\Gamma-P_y,
\end{align} 
with $\pi_x(w)=\Pi w$, $\pi_u(w)=\Gamma w$. 
Given solvability of~\eqref{eq:reg_lin} and the input-output stage cost $\ell$~\eqref{eq:ell_reg}, Assumption~\ref{ass:stab} and Assumption~\ref{ass:detect} reduce to stabilizability of $(A,B)$ (cf. Prop.~\ref{prop:stab}) and detectability of $(A,C)$  (cf. Prop.~\ref{prop:detect}).  
Satisfaction of Assumption~\ref{ass:detect} for $R\succ 0$ and $(A,C)$ detectable has also been shown in~\cite[Corollary~2]{hoger2019relation}. 
In case of polytopic constraints $\mathcal{Z}^p$ the MPC optimization problems in Sections~\ref{sec:reg}--\ref{sec:input} reduce to standard quadratic programs (QPs).

\subsection{Nonresonance  condition}
Consider the case where the matrix $S$ has only eigenvalues $\lambda$  of the form $\lambda=e^{2\pi i k/T}$ with some period length $T\in\mathbb{N}$ (Ass.~\ref{ass:exo_periodic}), which encompasses constant and sinusoidal exogenous signals $w$. 
Correspondingly, all the eigenvalues are on the unit circle, i.e., $|\lambda|=1$, as is standard in the literature~\cite[(A1)]{castillo1993nonlinear}, \cite[H1]{isidori1990output}. 
For simplicity, we consider square systems, i.e., $m=p$. 
To characterize the transmission zeros of a linear transfer matrix, we use Rosenbrock's system matrix
\begin{align}
G(\lambda):=\begin{pmatrix}
A-\lambda I_{n_p}&B\\
C&D
\end{pmatrix}.
\end{align}
In particular, $\lambda\in\mathbb{C}$ is a zero of the transfer matrix if the matrix $G(\lambda)$ does not have full rank, compare, e.g.,~\cite{davison1974properties}. 
The classical nonresonance condition (cf.~\cite[Lemma~4.1]{isidori2017lectures}) reduces to $\text{rank}(G(\lambda_k))=n_p+m$ for all $\lambda_k$ which are eigenvalues of $S$, i.e., the transmission zeros of the plant do not coincide with the poles of the exosystem. 
Solvability of~\eqref{eq:reg_lin} can be ensured if this \textit{nonresonance condition} holds and the matrices $\Pi$, $\Gamma$ are even unique since $m=p$, compare~\cite[Lemma~4.1]{isidori2017lectures}. 
Hence, Assumption~\ref{ass:regulator} holds if $\mathcal{Z}^p=\mathbb{R}^{n_p+m}$ and $\text{rank}(G(\lambda))=n_p+m$ for all $\lambda_k$ which are eigenvalues of $S$.

The following proposition shows that if the \textit{nonresonance condition} holds for all $T$-periodic exosystems, then results in Prop.~\ref{prop:non_res} remain valid, i.e., the augmented plant is detectable. 
\begin{proposition}
\label{prop:periodic_IOSS_linear}
Consider a square linear system, with $(A,C)$ detectable and $\text{rank}(G(\lambda_k))=n_p +m$ for all $\lambda_k=e^{2 i k\pi/T}$, $k=0,\dots,T-1$. 
Then the augmented plant~\eqref{eq:periodic_model}  is detectable.  
\end{proposition}
\begin{proof}
The augmented plant~\eqref{eq:periodic_model} corresponds to
\begin{align*}
{A}_a=&\begin{pmatrix}
A&B  E_2^\top\\
0&E_0
\end{pmatrix},~
{C}_a=\begin{pmatrix}
C&DE_2^\top\\0&0
\end{pmatrix}.
\end{align*} 
Detectability of $({A}_a,{C}_a)$ is equivalent to
\begin{align}
\label{eq:resonance_periodic_temp}
\text{rank}
\begin{pmatrix}
A-\lambda I_{n_p}&B E_2^\top\\
0&E_0-\lambda I_{mT}\\
C&D E_2^\top
\end{pmatrix}=n+Tm,
\end{align}
for all $\lambda\in\mathbb{C}$, which are eigenvalues of ${A}_a$ and satisfy $|\lambda|\geq 1$. 
First, consider $\lambda_k=e^{2\pi i k/T}$, in which case rank$(E_0-\lambda_k I_{mT})=m(T-1)$.
W.l.o.g. consider $m=1$. There exists one eigenvector $\xi^k=(e^{-2\pi i k/T},\dots e^{-2\pi i kT/T})$ satisfying $(E_0-\lambda_k I)\xi^k=0$. 
Furthermore, we have $E_2^\top \xi^k=e^{-2\pi i kT/T}=1$. Thus, the rank condition~\eqref{eq:resonance_periodic_temp} is equivalent to 
$\text{rank}(G(\lambda_k))=n_p+m$. 
For $\lambda_k\neq e^{2\pi i k/T}$, the rank condition reduces to detectability of $(A,C)$ and thus $(A_a,C_a)$ is detectable.
\end{proof}
We need to consider $\lambda_k=e ^{2ik\pi/T}$  for $k=0,\dots,T-1$ instead of only the eigenvalues of $S$ (cf. \cite[Lemma~4.1]{isidori2017lectures}), since in the incremental input regularization in Section~\ref{sec:input} we only use the fact that $S$ is $T$-periodic, but do not use the explicit eigenvalues of $S$ in the design. 
In case of redundant inputs $m>p$, we may be able to achieve the same output trajectory $y$ with different input trajectories $u$ and thus the augmented plant is not detectable. 

The following proposition shows that the dissipation-based characterization in Assumption~\ref{ass:nonres_nonlin} is equivalent to the rank condition in Proposition~\ref{prop:periodic_IOSS_linear}.
\begin{proposition}
\label{prop:linear_equiv}
Consider a square linear system with $(A,C)$ detectable. 
Assumption~\ref{ass:nonres_nonlin} holds if and only if $\text{rank}(G(\lambda_k))=n_p+m$ for all $\lambda_k=e^{2 i k \pi/T}$, $k=0,\dots ,T-1$. 
\end{proposition}
\begin{proof}
\textbf{Part I. }
Suppose Ass.~\ref{ass:nonres_nonlin} holds, but there exists some $\lambda_k$, such that $\text{rank}(G(\lambda_k))<n_p+m$, i.e., there exists some (complex) vector $(x^p,u)\neq 0$, such that $Ax^p+Bu=\lambda_kx^p$, $Cx^p+Du=0$. 
This corresponds to the existence of a $T-$periodic state and input trajectory $(x^p_t,u_t)$, which satisfies $Cx^p_t+Du_t=0$.
The periodicity of this trajectory implies that the augmented plant input satisfies ${u}^a_t=0$. 
Without loss of generality, consider $(z^p_t,v_t,{v}^a_t)=0$. 
Plugging the trajectories in Condition~\eqref{eq:nonlin_nonresonance_2} and using a telescopic sum we arrive at 
$\sum_{t=0}^{k-1}\|u_t\|^2\leq V_R(x_0,0)-V_R(x_k,k)\leq V_R(x_0,0)$.
Since $u_t$ periodic and the sum is upper bounded, this immediately implies $u\equiv 0$. 
Finally, since $y\equiv 0$ and $(A,C)$ detectable, $x^p\equiv 0$. Thus, the only periodic solution that satisfies $y=0$ is the trivial solution $(x^p,u)=0$ and thus $\text{rank}(G(\lambda_k))=n_p+m$. \\
\textbf{Part II. } Suppose $\text{rank}(G(\lambda_k))=n_p+m$ and $(A,C)$ is detectable. 
Then Prop.~\ref{prop:periodic_IOSS_linear} ensures that $({A}_a,{C}_a)$ is detectable and thus (cf.~\cite{cai2008input}) there exists a quadratic i-IOSS Lyapunov function $V_{a,o}({x}^{ap},z^{ap},w)=\|{x}^{ap}-{z}^{ap}\|_{{P}_{a,o}}^2$ satisfying Assumption~\ref{ass:IOSS}. 
Consider $V_R({x}^{ap},z^{ap},w)=c V_{a,o}({x}^{ap},z^{ap},w)$ with $c=\frac{2}{(1-\rho_{a,o})\lambda_{\min}({P}_{a,o})}>0$.
Inequality~\eqref{eq:nonlin_nonresonance_1} holds with $c_{R,u}:=c\cdot \lambda_{\max}({P}_{a,o})$. 
The definition of the input $u$ in~\eqref{eq:periodic_model_4} implies
\begin{align}
\label{eq:input_quad_bound}
&\|u-v\|^2=\|E_2^\top (\xi-\Xi)+{u}^a-{v}^a\|^2\\
\leq &2\|{u}^a-{v}^a\|^2+2\underbrace{\|E_2E_2^\top \|}_{=1}\|\xi-\Xi\|^2.\nonumber
\end{align}
Thus, condition~\eqref{eq:nonlin_nonresonance_2} holds with $c_R=\max\{c\cdot c_{o,2},c\cdot (c_{o,1}+c_{o,2})+2\}>0$ using
\begin{align*}
&\|A_a({x}^{ap}-z^{ap})+{B}_a ({u}^a-v^a)\|_{P_{a,o}}^2-\|x^{ap}-z^{ap}\|_{P_{a,o}}^2\\
\leq &-c  (1-\rho_{a,o})\|{x}^{ap}-{z}^{ap}\|_{P_{a,o}}^2+c  (c_{o,1}+c_{o,2})\|{u}^a-{v}^a\|^2\\
&+c\cdot c_{o,2}\|{C}({x}^p-{z}^p)+{D}({u}-{v})\|^2\\
\leq &-2\|\xi-\Xi\|^2+(c_R-2)\|{u}^a-{v}^a\|^2\\
&+c_R\|h(x^p,u,w)-h(z^p,v,w)\|^2 \\
\stackrel{\eqref{eq:input_quad_bound}}{\leq} &c_R(\|h^a(x^{ap},u^a,w)-h^a(z^{ap},v^a,w)\|^2 
+\|{u}^a-{v}^a\|^2)\\
&-\|u-v\|^2.
&\qedhere
\end{align*}
\end{proof}
We point out again that the rank condition $\text{rank}(G(\lambda_k))=n_p+m$ is similar to the eigenvalue and rank conditions used for input observability/detectability in~\cite[Thm.~2--3]{hou1998input}, which is a closely related problem. A crucial relaxation in the considered problem is that only periodic inputs need to be observable/detectable and thus the rank condition only needs to be checked for the eigenvalues $\lambda_k$ corresponding to the period length $T$, as opposed to all eigenvalues (outside the unit disc).

\subsection{Minimum phase - stable zeros}
In the following, we consider a SISO system as in Section~\ref{sec:minPhase}. 
The relative degree $d\in\mathbb{N}$ in Assumption~\ref{ass:BINF} corresponds to $CA^kB=0$, $k=0,\dots,d$ and $CA^{d+1}B\neq 0$  and the maps $\Phi,\tilde{\Phi}$ are linear. 
Furthermore, the zero dynamics are always well-defined (Ass.~\ref{ass:zero}), using $\alpha(x)=K_\alpha x$ with $K_\alpha=-\frac{CA^{d+2}}{CA^{d+1}B}$.
If we consider the closed-loop system $u=K_\alpha x+\Delta{u}$ we have $\eta^+=A_{\eta}\eta+A_{\eta,z}z+B_{\eta,u}\Delta{u}+B_{\eta,w}w$ and the zero dynamics are stable if $A_{\eta}$ is Schur. 
In this case, the dynamics in $\eta$ are obviously also ISS with a quadratic function $V_\eta=\|\eta-{\eta}_w\|_{P_\eta}^2$ and thus Assumption~\ref{ass:minPhase} holds. 
The eigenvalues of $A_{\eta}$ characterizing the zero dynamics correspond to the zeros of the transfer function (assuming $(A,B,C,D)$ corresponds to a minimal realization), compare, e.g.,~\cite{davison1974properties,isidori2013zero}.

\subsection{Summary}
Suppose we have a linear square system that satisfies
\begin{itemize}
\item $(A,B)$ stabilizable, $(A,C)$ detectable,
\item Eigenvalues of $S$ satisfy $\lambda_k=e^{2\pi i k/T}$, $T\in\mathbb{N}$,
\item Nonresonance  condition: $\text{rank}(G(\lambda_k))=n_p+m$, $\forall\lambda_k=e^{2\pi i k /T}$, $k=0,\dots,T-1$,
\item No constraints: $\mathcal{Z}^p=\mathbb{R}^{n_p+m}$.
\end{itemize}
Then Assumptions~\ref{ass:regulator}--\ref{ass:IOSS} hold and  the MPC schemes based on~\eqref{eq:MPC} and \eqref{eq:MPC_input} both solve the output regulation problem for $N$ sufficiently large. 
Furthermore, in case $\text{rank}(G(\lambda))=n_p+m$ for all $\lambda\in\mathbb{C}$ satisfying $|\lambda|\geq 1$ (stable zeros, Ass.~\ref{ass:minPhase}), also the MPC scheme in Section~\ref{sec:minPhase} solves the output regulation problem for $N$ sufficiently large.
In case the joint system $(x^p,w)$ is detectable, we can design a stable observer and implement an error feedback MPC with noisy output measurements that ensures finite-gain $\mathcal{L}_2$ stability.

In the linear case, we clearly see that the considered conditions align with the typical assumptions employed to solve the output regulation problem, compare, e.g., the necessary and sufficient conditions in~\cite[Thm.~2]{davison1976robust}.

\section{Numerical example}
\label{sec:num}
In Section~\ref{sec:num_1} we demonstrate by means of an academic example the relevance of using an  input regularization  for systems with unstable zero dynamics. 
Then, Section~\ref{sec:num_2} shows the applicability of the proposed approach at the example of \textit{offset-free} MPC with error feedback (Rk.~\ref{rk:error_feedback}) and nonlinear dynamics. 
 The online computations are done with Matlab using IPOPT within CasADi~\cite{andersson2019casadi}.
 %
\subsection{Academic linear example - non-minimum phase}
\label{sec:num_1}
Consider the following academic linear system $x_{t+1}=0.5x_t+u_t$, $y_t=x_t-u_t$.
This system is stable and has a direct feed through.  
For simplicity, suppose we have no constraints, i.e., $\mathbb{X}\times\mathbb{U}=\mathbb{R}^{n+m}$. 
The solution to the MPC optimization problem in Section~\ref{sec:minPhase} with the output stage cost $\ell_y=\|y\|^2=\|x-u\|^2$ satisfies $u^*_{0|t}=x_t$, $x_{t+1}=1.5 x_t$, $y_t=0$ and $V_N(x_t)=0$ for any horizon $N\in\mathbb{N}$ and all $t\geq 0$. 
Thus, the MPC scheme minimizes the output $\|y\|$, but the resulting state and input trajectory is unstable. 
This problem is inherently related to the \textit{singular input cost} $\ell_y$ and the existence of unstable zero dynamics. 
A similar phenomenon appears in high-gain controllers which quickly ensure $\lim_{t\rightarrow\infty}\|y_t\|=0$ and thus often fail to stabilize systems with unstable zero dynamics (non-minimum-phase), compare~\cite[Sec.~3.4]{davison1974properties}. 
If the system is subject to compact input constraints $u_t\in\mathbb{U}$ and has an arbitrarily small  non-zero initial condition $x_0$, the closed-loop error $\|y_t\|$ would be zero for $t\in[0,K]$, until  the input $u_t=x_t\notin\mathbb{U}$ and then the error $y$ becomes nonzero. 
This demonstrates that for general non-minimum-phase systems, an input regularization as used in Sections~\ref{sec:reg} or \ref{sec:input} is needed to ensure stability.

If we use the incremental input penalty from Section~\ref{sec:input} with $\ell_a=\|y\|^2+\|\Delta u\|^2$, we can invoke Corollary~\ref{corol:input} to ensure exponential stability for a sufficiently large horizon $N$.
In particular, for the augmented state $x^a_t=(x_t,u_{t-1})$ and input $u^a=\Delta u$, we define $\sigma^a=\|x^a\|_P^2$ with $P$ from the algebraic Riccati equation corresponding to the LQR.
Thus,  Assumption~\ref{ass:stab} is naturally satisfied with $\gamma_s=1$. 
Assumption~\ref{ass:detect} holds with $W=0$, $\gamma_o=0$, $\epsilon_o=0.3343$, which can be computed based on~\eqref{eq:detect_grimm_2} as a generalized eigenvalue. 
Thus, Corollary~\ref{corol:input} is applicable for $N>N_1=1+(\gamma_s/\epsilon_o)^2\approx 9.9$.
If we further use the improved bound discussed in Remark~\ref{rk:observable}, we can ensure stability for $N>N_{\overline{Y},s}\approx 3.3$ (cf.~Prop.~\ref{prop:less_cons}).

\subsection{Nonlinear offset-free tracking}
\label{sec:num_2}
In the following example, we demonstrate the applicability of the proposed approach to nonlinear offset free MPC using both, the pure output tracking formulation from Section~\ref{sec:minPhase} and the incremental input formulation from Section~\ref{sec:input}.
We consider the following nonlinear model of a cement milling circuit taken from~\cite{magni2001output}:
\begin{subequations}
\label{eq:example_model_offset}
\begin{align}
0.3\dot{x}_1=&-x_1+(1-\alpha(x_2,u_2))\phi(x_2),\\
\dot{x}_2=&-\phi(x_2)+u_1+x_3,\\
0.01\dot{x}_3=&-x_3+\alpha(x_2,u_2)\phi(x_2),\\
\phi(x_2)=&\max\{0,-0.1116\cdot x_2^2+16.50x_2\},\\
\alpha(x_2,u_2)=&\dfrac{\phi^{0.8}(x_2)\cdot u_2^4}{3.56\cdot 10^{10}+\phi(x_2)^{0.8}\cdot u_2^4}.
\end{align}
\end{subequations}
with $x\in\mathbb{R}^3$, $u\in\mathbb{R}^2$. 
The discrete-time model is computed using the $4$th order Runge Kutta method and a sampling time of one minute\footnote{%
The units in equations~\eqref{eq:example_model_offset} are hours.}. 
The system is subject to compact input constraints $u\in\mathbb{U}=[80,150]\times[165,180]$ and no state constraints $\mathbb{X}=\mathbb{R}^n$.
The error is given by $y=(x_1,x_3)-(w_1,w_2)$, where $w$ corresponds to the constant output reference $w=(110,425)$.

In the following, we briefly show that the considered assumptions hold on the subset $x_2\in[45,55]$, $w\in\mathbb{W}= [100,120]\times[410,430]$, which provides a sufficiently large region of attraction. 
First, the unique solution to the regulator equations (Ass.~\ref{ass:regulator}) can be analytically computed as 
\begin{align*}
\pi_x(w)=& [w_1,~73.9-\sqrt{5.5\cdot 10^3-8.9(w_1+w_2)},~w_2]^\top,\\
\pi_u(w)=&[w_1,~434\left(\frac{w_2}{w_1}\right)^{0.25}(w_1+w_2)^{-0.2}]^\top,
\end{align*}
which is also Lipschitz continuous on the considered region. 
The system is open-loop incrementally stable (and hence satisfied Assumptions~\ref{ass:increm_stab}, \ref{ass:stab}), which we verified numerically by computing (via gridding) a constant contraction metric~\cite{manchester2017control} (which corresponds to a quadratic incremental Lyapunov function $V_s$). 
Correspondingly, the system also trivially satisfies the detectability condition (Ass.~\ref{ass:IOSS}) with $V_o=V_s$. 
Furthermore, one can show that the system is flat and contains no zero-dynamics. 
Hence, the conditions regarding the minimum phase property (Ass.~\ref{ass:BINF}--\ref{ass:minPhase}) are trivially satisfied. 
Similarly, the nonresonance condition (Ass.~\ref{ass:nonres_nonlin}) follows due to the absence of zero-dynamics (cf. the discussion in Section~\ref{sec:input_3}) and a corresponding quadratic incremental storage function $V_R$ can be computed similar to~\cite{sanfelice2012metric,verhoek2020convex}. 
Hence, we have shown that all the considered assumptions hold. However, due to the complexity of the system the resulting bounds on the sufficiently long prediction horizon $N$ from the derived theorems are too conservative to be applied. 
Thus, we simply implement the two MPC schemes (Sec.~\ref{sec:minPhase}/\ref{sec:input}) with $N=6$, $Q=I_2$ and $R=0/R=10^{-2}\cdot I_2$.

The resulting closed loop for $x^p_0=(120,55,450)$ can be seen in Figure~\ref{fig:Offset_state}.
Both proposed MPC formulations smoothly track the output reference, while satisfying the active input constraints. 
If we compare the MPC formulation with and without input regularization (Sec.~\ref{sec:minPhase}/\ref{sec:input}), the resulting closed-loop state trajectories are almost indistinguishable, while the absence of input regularization leads to more aggressive control inputs. 
\begin{figure}[hbtp]
\begin{center}
\includegraphics[width=0.43\textwidth]{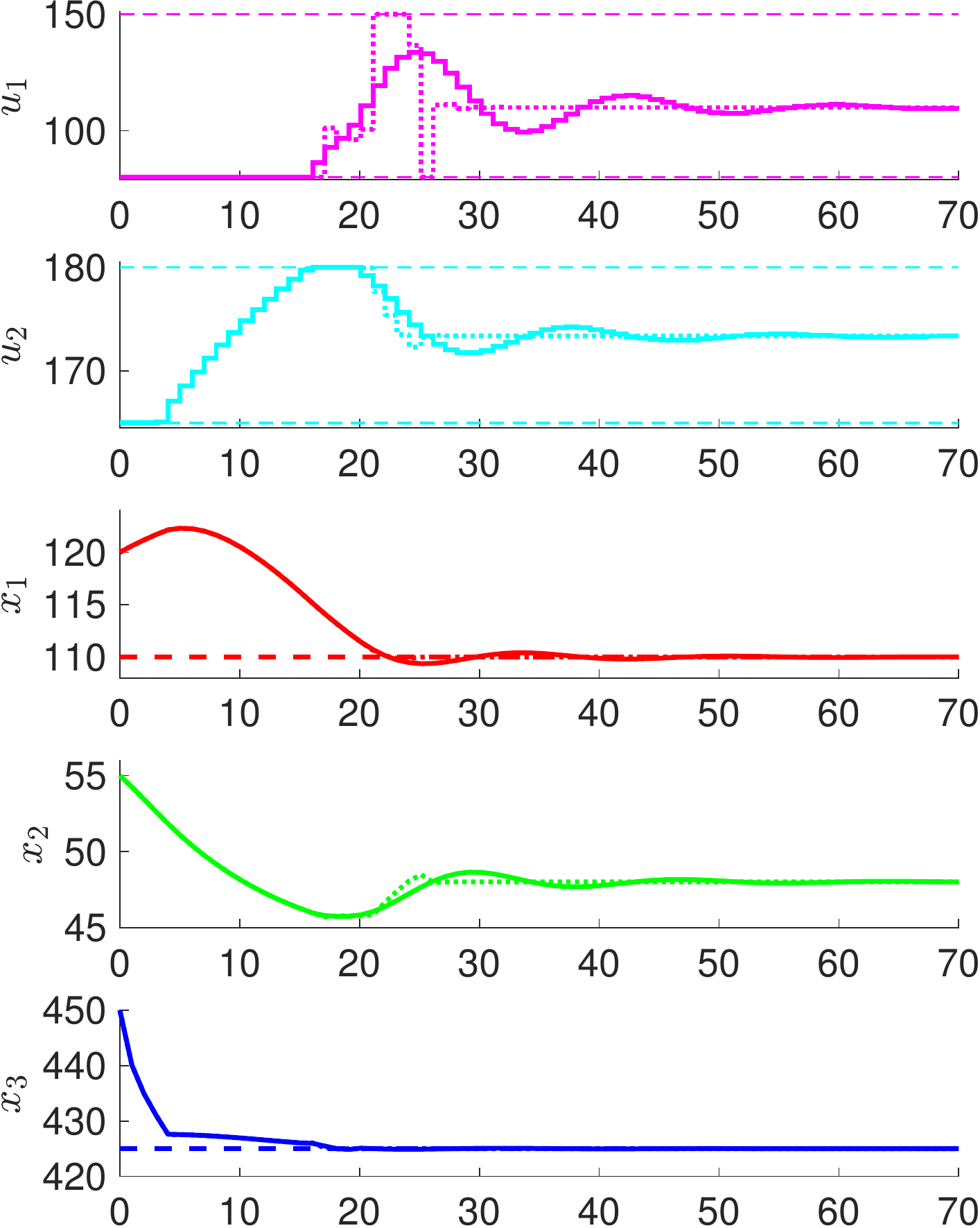}
\end{center}
\caption{Offset free tracking with state feedback: Incremental input regularization ($R=10^{-2}$, solid) and without input regularization ($R=0$, dotted). 
Output reference $(w_1,w_2)=(110,425)$ and input constraints are dashed.}
\label{fig:Offset_state}
\end{figure}

\subsubsection*{Noisy error feedback and inherent robustness}
Now we consider more the realistic scenario of noisy error feedback as discussed in Remark~\ref{rk:error_feedback}, i.e., only noisy output measurements $\tilde{y}=y+\eta$ are available, with $\eta$ uniformly distributed in $[-1,1]^2$. 
As in~\cite{magni2001output}, we design an extended Kalman filter (EKF) as an observer and implemented the output regulation MPC in a certainty equivalent fashion, compare Appendix~\ref{app:error_feedback} for details. 
The EKF uses an initial variance of $\Sigma=100\cdot I_n$ and unit variance for noise and disturbances in the design.
The initial state estimate is given by $\hat{x}_0=(\hat{x}^p_0,\hat{w}_0)=(100,50,400,100,400)$. 
The resulting closed loop can be seen in Figure~\ref{fig:Offset_error_state}. 
We can see that for both MPC formulations the control performance is rather insensitive to the noise and estimation error.
The resulting closed-loop state trajectories for the two MPC formulations are almost indistinguishable, while the absence of input regularization leads to more aggressive control inputs, especially in $u_1$.
Even though the tracking error is almost zero at the end of the simulation time, the observer error $\hat{x}^p_3-x^p_3$ and $\hat{w}_2-w_2$ is of the order $10$ and requires significantly longer to converge close to the origin.

\begin{figure}[hbtp]
\begin{center}
\includegraphics[width=0.43\textwidth]{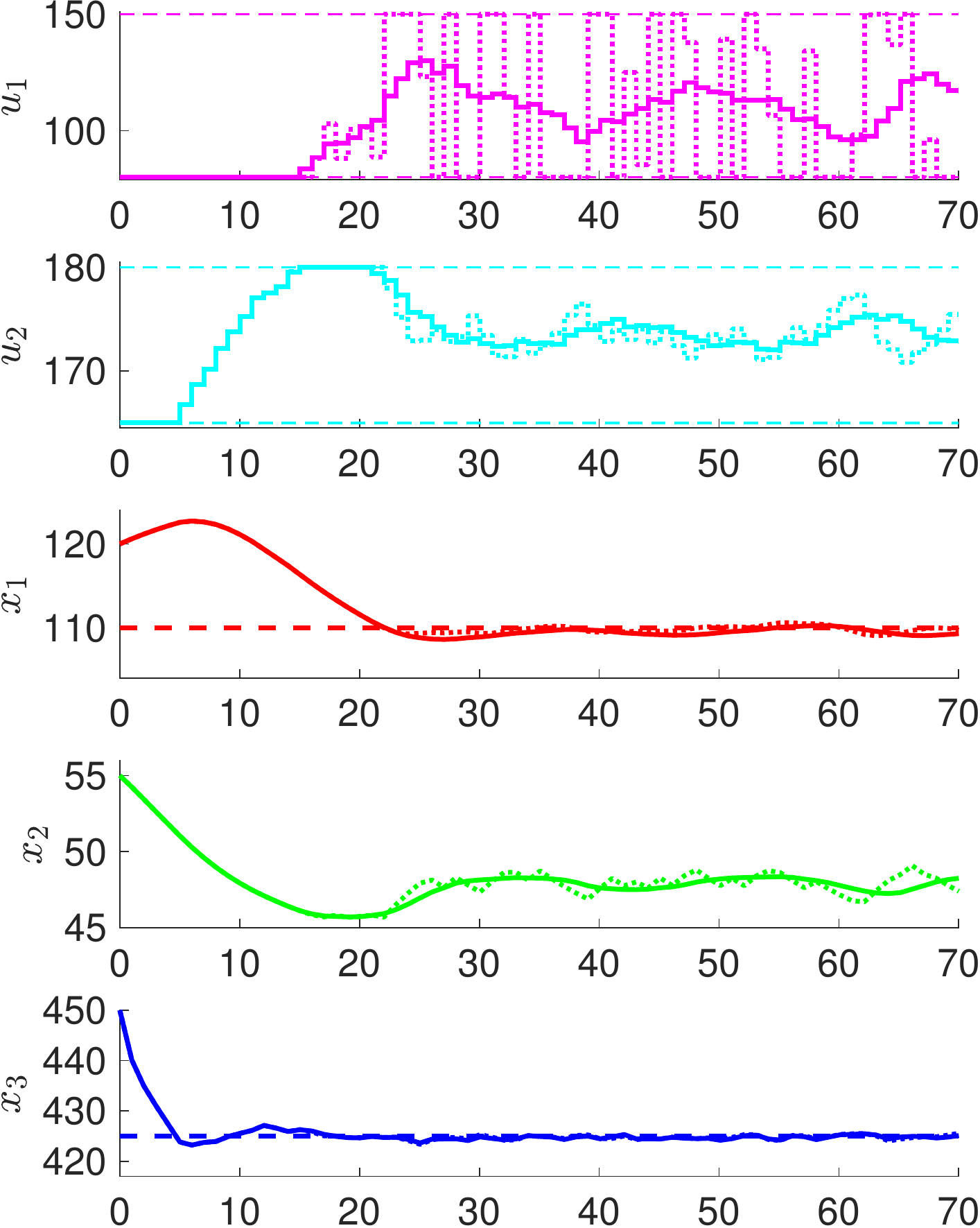}
\end{center}
\caption{Offset free tracking with noisy error feedback: Incremental input regularization ($R=10^{-2}$, solid) and without input regularization ($R=0$, dotted). 
Output reference $(w_1,w_2)=(110,425)$ and input constraints are dashed.}
\label{fig:Offset_error_state}
\end{figure}

\subsubsection*{Message}
The main benefits of the proposed approach is its simplicity. 
For the implementation, we only require a prediction model, a user can suitable tune input and output weights $Q,R$, and a sufficiently long prediction horizon $N$ needs to be chosen. 
In case of error feedback, we additionally need to design a stable observer, e.g., here an extended Kalman filter.
Most importantly, the proposed design did not require any complex offline computations.
This is in contrast to most approaches for output regulation (cf. e.g.~\cite{castillo1993nonlinear,pavlov2006uniform,magni2005solution,falugi2013tracking}), that typically first need to compute a solution to the regulator equations~\eqref{eq:reg}, which is in general non-trivial. 
Furthermore, compared to classical approaches to output regulation (cf., e.g.,~\cite{castillo1993nonlinear,pavlov2006uniform}), the proposed approach offers a large region of attraction despite the presence of hard input constraints. 

Finally, compared to tracking MPC formulations~\cite{limon2018nonlinear,kohler2020nonlinear}, the proposed approach has the following advantages:
\begin{enumerate}[label=(\alph*)]
\item No complex offline design for terminal ingredients,
\item No feasibilities issues and strong stability properties in the noisy error feedback case due to the absence of terminal constraints,
\item A larger region of attraction,
\item No additional decision variables to compute the optimal steady-state $x=\pi_{x}(w)$ online.
\end{enumerate}
The main drawbacks compared to tracking MPC formulations~\cite{limon2018nonlinear,kohler2020nonlinear} is the fact that potentially a larger prediction horizon $N$ may be needed to guarantee stability and that guaranteed performance in case of unreachable trajectories (Ass.~\ref{ass:regulator} does not hold) are difficult to establish.
 
\section{Conclusion}
\label{sec:sum}
We have  presented an MPC framework that solves the nonlinear constrained output regulation problem, given suitable stabilizability and detectability conditions and a sufficiently long prediction horizon. 
In particular, we have presented two MPC formulations (with/without input regularization, Sec.~\ref{sec:minPhase}/\ref{sec:input}), that are suitable for minimum phase systems and periodic exogenous signals, respectively. Both MPC schemes do \textit{not} require a solution to the regulator equations or other complex offline designs.
We have demonstrated the applicability and simplicity of the MPC formulations with a numerical example involving nonlinear offset-free tracking and noisy error feedback. 

Future research focuses on achieving \textit{robust} output regulation.

\bibliographystyle{IEEEtran}  
\bibliography{Literature}

\begin{IEEEbiography}[{\includegraphics[width=1in,height=1.25in,clip,keepaspectratio]{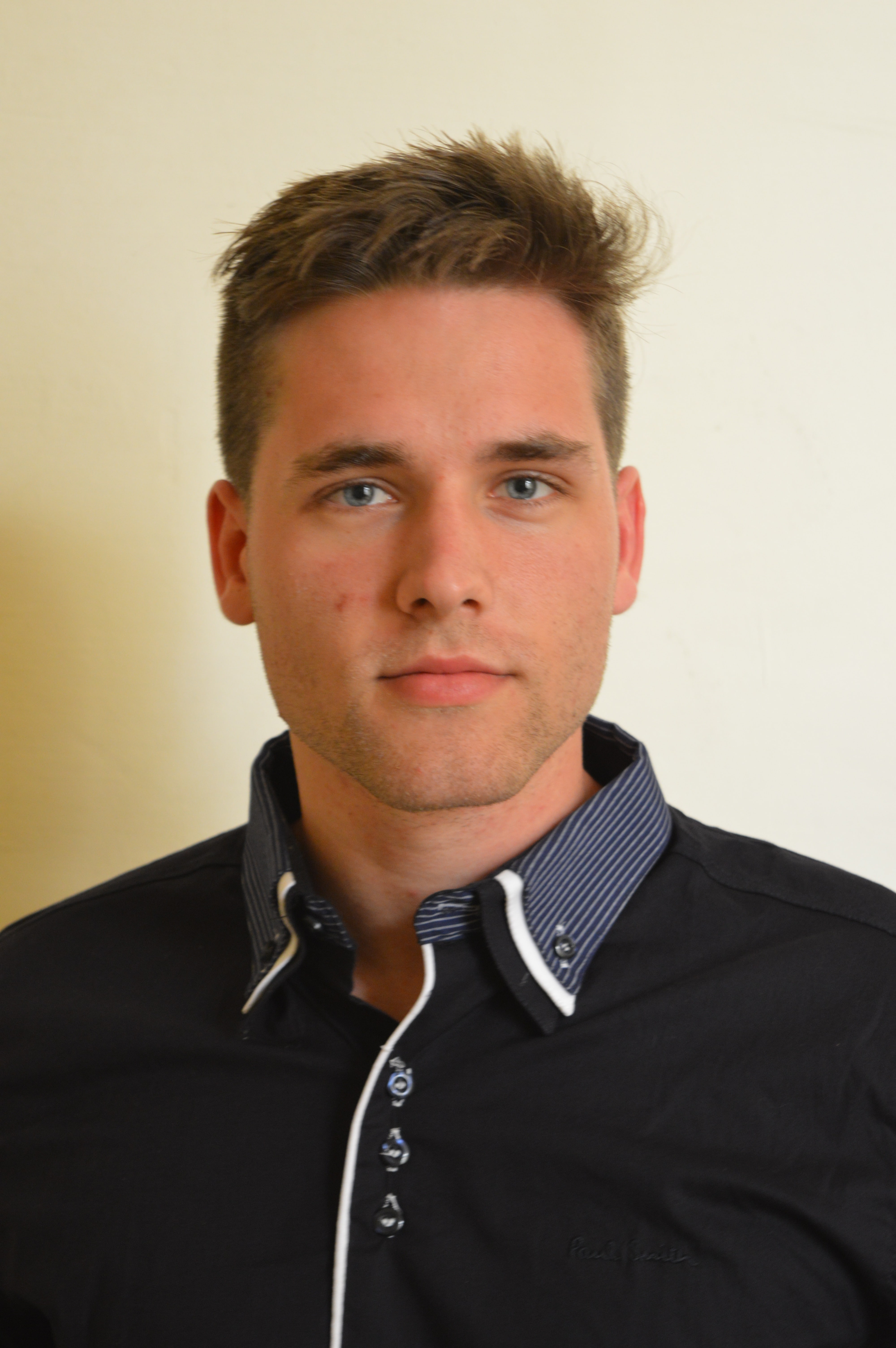}}]{Johannes K\"ohler}
 received his Master degree in Engineering Cybernetics from the University of Stuttgart, Germany, in 2017.
He has since been a doctoral student at the \emph{Institute for Systems Theory and Automatic Control} under the supervision of Prof. Frank Allg\"ower and a member of the International Research Training Group (IRTG) "Soft Tissue Robotics" at the University of Stuttgart. 
His research interests are in the area of model predictive control and nonlinear uncertain systems. 
\end{IEEEbiography}


\begin{IEEEbiography}[{\includegraphics[width=1in,height=1.25in,clip,keepaspectratio]{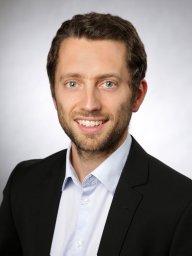}}]{Matthias A. M\"uller}
 received a Diploma degree in Engineering Cybernetics from the University of Stuttgart, Germany, and an M.S. in Electrical and Computer Engineering from the University of Illinois at Urbana-Champaign, US, both in 2009.
In 2014, he obtained a Ph.D. in Mechanical Engineering, also from the University of Stuttgart, Germany, for which he received the 2015 European Ph.D. award on control for complex and heterogeneous systems. Since 2019, he is director of the Institute of Automatic Control and full professor at the Leibniz University Hannover, Germany. 
He obtained an ERC Starting Grant in 2020 and is recipient of the inaugural Brockett-Willems Outstanding Paper Award for the best paper published in Systems \& Control Letters in the period 2014-2018. His research interests include nonlinear control and estimation, model predictive control, and data-/learning-based control, with application in different fields including biomedical engineering.
\end{IEEEbiography}
%
\begin{IEEEbiography}[{\includegraphics[width=1in,height=1.25in,clip,keepaspectratio]{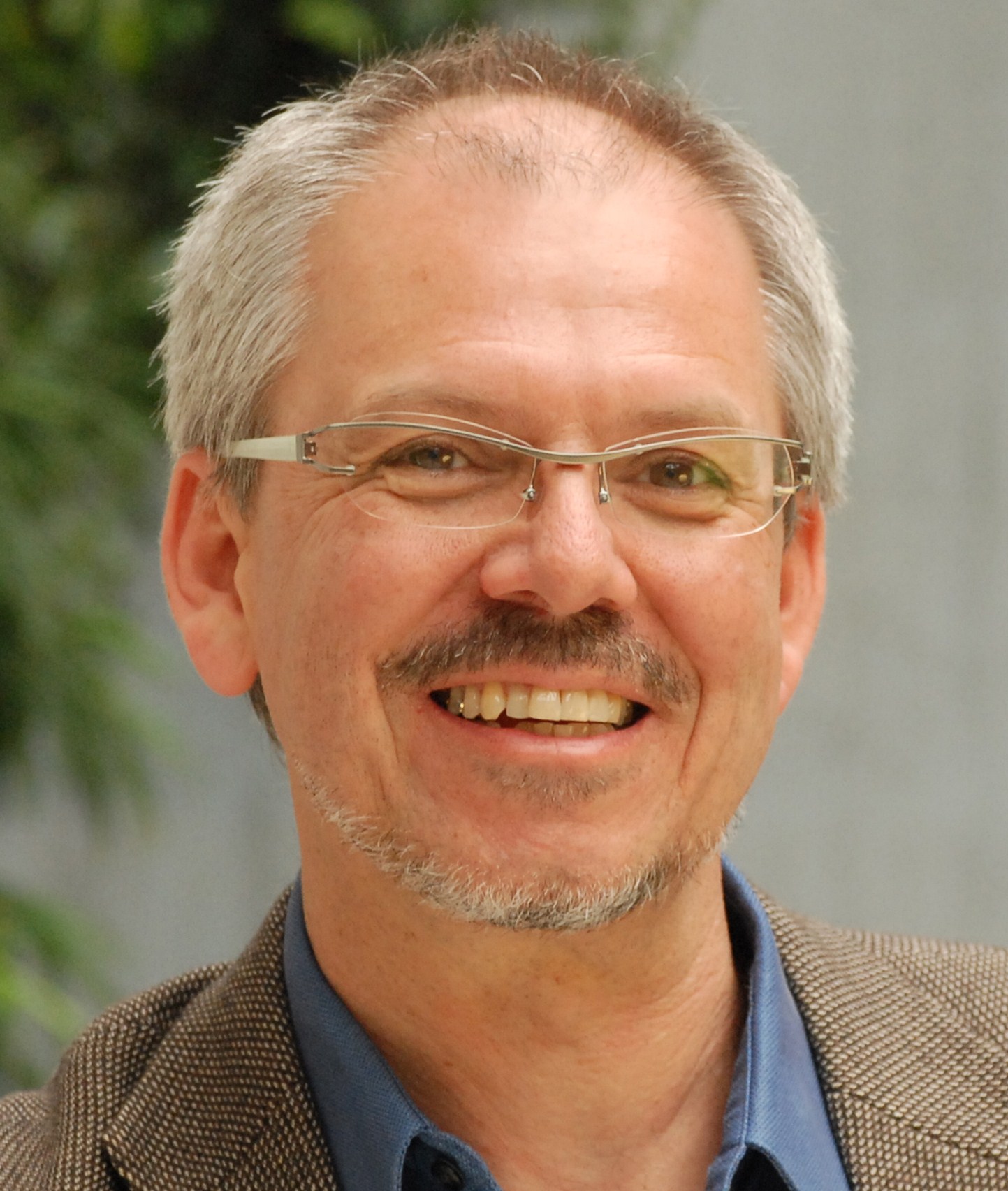}}]{Frank Allg\"ower}
 studied Engineering Cybernetics and Applied Mathematics in Stuttgart and at the University of California, Los Angeles (UCLA), respectively, and received his Ph.D. degree from the University of Stuttgart in Germany. 
Since 1999 he is the Director of the \emph{Institute for Systems Theory and Automatic Control} and professor at the University of Stuttgart.
His research interests include networked control, cooperative control, predictive control, and nonlinear control with application to a wide range of fields including systems biology.
For the years 2017-2020 Frank serves as President of the International Federation of Automatic Control (IFAC) and since 2012 as Vice President of the German Research Foundation DFG.
\end{IEEEbiography}

\clearpage
\appendix
In Appendix~\ref{app:obs}, we show that the performance estimates in Theorem~\ref{thm:MPC}/Corollary~\ref{corol:main}, which are the basis of the proposed MPC framework, can be improved by using an additional observability condition.    
In Appendix~\ref{app:error_feedback}, we extend the results in Sections~\ref{sec:reg}--\ref{sec:input} to the noisy error feedback case.
\subsection{Less conservative bounds using observability}
\label{app:obs}
In the following, we extend the general analysis in Theorem~\ref{thm:MPC} and hence also from~\cite{grimm2005model} to provide less conservative bounds, given an additional observability condition.
The suboptimality index $\alpha_N$~\eqref{eq:def_alpha} in Theorem~\ref{thm:MPC} decreases with $\gamma^2/N$ {(we consider $\gamma\approx \gamma_s\approx \gamma_{\overline{Y}}$ in this discussion to keep in line with the notation in~\cite{grune2012nmpc})}
 and thus the sufficient horizon $N_{\overline{Y}}$ scales quadratically with $\gamma$, similar to~\cite[Variant 1]{grune2012nmpc}. 
However,  under the assumption that $\ell$ is positive definite (Ass.~\ref{ass:detect} holds with $W=0$) the derivations in~\cite{tuna2006shorter} and \cite[Variant 2]{grune2012nmpc}, provide bounds where $N_{\overline{Y}}$ scales with $2\gamma\log \gamma$. 
In the following, we show how we can obtain bounds for $N_{\overline{Y}}$, $\alpha_N$ similar to~\cite[Variant 2]{grune2012nmpc} using an additional/stronger \textit{observability} condition.
Proposition~\ref{prop:less_cons} derives the improved bounds given Assumption~\ref{ass:obs}.
Proposition~\ref{prop:obs_plant} shows that Assumption~\ref{ass:obs} follows naturally from observability.
\begin{assumption}
\label{ass:obs}
(Observability)
There exist constants $\nu\in\mathbb{N}$ and $c_o>0$, such that any trajectory satisfying $(x_t,u_t)\in\mathcal{Z}$, $x_{t+1}=f(x_t,u_t)$ for all $t\geq 0$ satisfies the following bound
\begin{align}
\label{eq:obs}
\sigma(x_{t+\nu})\leq c_o\sum_{k=0}^{\nu -1} \ell(x_{t+k},u_{t+k}),~\forall t\geq 0.
\end{align}
\end{assumption}
In the simple case that $\sigma(x)=\|x\|^2$ and $\ell(x,u)\geq \|x\|^2+\|u\|^2$, Ass.~\ref{ass:obs} holds with $\nu=1$ if $f$ is Lipschitz and thus allows to directly use arguments from~\cite[Variant 2]{grune2012nmpc}.
\begin{proposition}
\label{prop:less_cons}
Let Assumptions~\ref{ass:stab}, \ref{ass:detect}, \ref{ass:obs} hold. 
For any constant $\overline{Y}>0$, there exists a constant $N_{\overline{Y},s}>0$, such that for $N>N_{\overline{Y},s}$ and  initial condition $x_0\in \mathbb{X}_{\overline{Y}}$, the closed loop satisfies $\alpha_{N,s}\sum_{t=0}^{\infty} \ell(x_t,u_t)\leq \overline{Y}$ with 
\begin{align}
\label{eq:alpha_less_cons}
\alpha_{N,s}:=&1-\dfrac{\gamma_{\overline{Y},s}\gamma_s c_o}{\epsilon_o} \left(\dfrac{\gamma_{\overline{Y},s}c_o}{\gamma_{\overline{Y},s}c_o+1}\right)^{N_\nu}>0,\\
N_\nu:=&\left\lfloor \frac{N-\nu}{\nu}\right\rfloor,\quad  \gamma_{\overline{Y},s}:=\max\{\gamma_s,\overline{Y}/\delta_s\}.\nonumber
\end{align}
\end{proposition}
\begin{proof}
\begin{subequations}
For any $x_t\in\mathbb{X}_{\overline{Y}}$ and any $k\in\{0,\dots,N-1\}$, $W,\ell\geq 0$  implies $V_{N-k}(x^*_{k|t})\leq V_N(x_t)\leq Y_N(x_t)\leq \overline{Y}$ and thus $V_{N-k}(x^*_{k|t})\leq \gamma_{\overline{Y},s}\sigma(x^*_{k|t})$, similar to Thm.~\ref{thm:MPC}. 
Abbreviate $\ell_k=\ell(x_{k|t}^*,u_{k|t}^*)$. 
Using observability (Ass.~\ref{ass:obs}), we obtain for any $p\in\{\nu,\dots,N-1\}$:
\begin{align}
\label{eq:inter_3}
\sum_{k=p}^{N-1}\ell_k\leq V_{N-p}(x^*_{p|t})\leq \gamma_{\overline{Y},s} \sigma(x_{p|t}^*)\stackrel{\eqref{eq:obs}}{\leq} \gamma_{\overline{Y},s}c_o\sum_{k=p-\nu}^{p-1}\ell_k.
\end{align}
Furthermore, we obtain
\begin{align}
\label{eq:inter_4}
&\sum_{k=p-\nu}^{N-1}\ell_k=\sum_{k=p-\nu}^{p-1}\ell_k+\sum_{k=p}^{N-1}\ell_k\\
\stackrel{\eqref{eq:inter_3}}{\geq}&\dfrac{1}{\gamma_{\overline{Y},s}c_o} \sum_{k=p}^{N-1}\ell_k+\sum_{k=p}^{N-1}\ell_k\nonumber
= \dfrac{\gamma_{\overline{Y},s}c_o+1}{\gamma_{\overline{Y},s}c_o}\sum_{k=p}^{N-1}\ell_k,
\end{align}
for any $p\in\{\nu,\dots,N-1\}$. 
Applying this inequality recursively, we obtain
\begin{align}
\label{eq:turnpike_exp}
&\dfrac{1}{c_o}\sigma(x^*_{N|t})\stackrel{\eqref{eq:obs}}{\leq} \sum_{k=N-\nu}^{N-1}\ell_k\leq \sum_{k=\nu N_\nu}^{N-1}\ell_k\\
\stackrel{\eqref{eq:inter_4}}{\leq}& \left(\dfrac{\gamma_{\overline{Y},s}c_o}{\gamma_{\overline{Y},s}c_o+1}\right)^{N_\nu} \sum_{k=0}^{N-1}\ell_k
\stackrel{\eqref{eq:MPC}}{\leq}  \left(\dfrac{\gamma_{\overline{Y},s}c_o}{\gamma_{\overline{Y},s}c_o+1}\right)^{N_\nu}V_N(x_t). \nonumber
\end{align}
Note that the bound~\eqref{eq:turnpike_exp} decays exponentially in $N$, compared to the bound in~\eqref{eq:turnpike}, which decays with $1/N$. 
For $N\geq N_{0,s}:=\nu\dfrac{\log(c_o \overline{Y}/\delta_s)}{\log(\gamma_{\overline{Y},s}c_o+1)-\log(\gamma_{\overline{Y},s}c_o)}+\nu$, this implies $x^*_{N|t}\in\mathbb{X}_\delta$. 
Correspondingly, we can use Assumption~\ref{ass:stab} to obtain
 \begin{align*}
&V_N(x_{t+1})-V_N(x_t)+\ell(x_t,u_t)\stackrel{\eqref{eq:stab_grimm}}{\leq} \gamma_s\sigma(x^*_{N|t})\\
\stackrel{\eqref{eq:turnpike_exp}}{\leq}&\gamma_s c_o \left(\dfrac{\gamma_{\overline{Y},s}c_o}{\gamma_{\overline{Y},s}c_o+1}\right)^{N_\nu} V_N(x_t)\\
\leq&\gamma_s c_o \gamma_{\overline{Y},s} \left(\dfrac{\gamma_{\overline{Y},s}c_o}{\gamma_{\overline{Y},s}c_o+1}\right)^{N_\nu}\sigma(x_t).
\end{align*}
Analogous to the derivation in Theorem~\ref{thm:MPC}, we obtain
\begin{align*}
Y_N(x_{t+1})-Y_N(x_t)\leq -\alpha_{N,s}\epsilon_o\sigma(x_t),
\end{align*}
with $\alpha_{N,s}$ according~\eqref{eq:alpha_less_cons} using~\eqref{eq:detect_grimm_2}.  
Finally, $\alpha_{N,s}>0$ and $N>N_{0,s}$ holds for
\begin{align*}
N>N_{\overline{Y},s}:=\nu\dfrac{\max\{\log(\gamma_{\overline{Y},s}\gamma_s c_o/\epsilon_o), \log(c_o\overline{Y}/\delta_s)\}}{\log(\gamma_{\overline{Y},s}c_o+1)-\log(\gamma_{\overline{Y},s}c_o)}+\nu.
\end{align*}
The remainder of the proof is analogous to Theorem~\ref{thm:MPC}. 
\end{subequations}
\end{proof} 
More recently, these bounds have been improved using a proof of contradiction to yield  $N_{\overline{Y},s}$ linear in $\bar{Y}$, compare~\cite{kohler2021dynamic}.
 \subsubsection*{Observability}
We consider the following observability condition, similar to standard conditions used in the literature for optimization based observer design~\cite[Def.~4.28]{rawlings2017model}. 
\begin{assumption}
\label{ass:obs_plant} ($\nu$-step i-OSS)
There exist constants $\nu\in\mathbb{N}$ and $c_{obs}>0$, such that for any trajectories satisfying 
$x^p_{t+1}=f^p(x^p_t,u_t,w_t)$, $z^p_{t+1}=f^p(z^p_t,u_t,w_t)$, $w_{t+1}=s(w_t)$, $(x^p_t,u_t)\in\mathcal{Z}^p$, $(z^p_t,u_t)\in\mathcal{Z}^p$ for all $t\geq 0$, the following inequality holds for all $t\geq 0$:
\begin{align}
\label{eq:observable}
&\|x^p_t-z^p_t\|^2\\
&\leq c_{\text{obs}}\sum_{j=0}^{\nu-1} \|h(x^p_{t+j},u_{t+j},w_{t+j})-h(z^p_{t+j},u_{t+j},w_{t+j})\|^2. \nonumber
\end{align}
\end{assumption}
Assumption~\ref{ass:obs_plant} implies that two trajectories subject to the same input/disturbance $u_t/w_t$ generate the same output $h$ over $\nu$ steps, only if they had the same initial condition.  
In the linear case (Sec.~\ref{sec:linear}), Assumption~\ref{ass:obs_plant} reduces to $(A,C)$ observable.
\begin{proposition}
\label{prop:obs_plant}
Let Assumptions~\ref{ass:regulator} and \ref{ass:obs_plant} hold and suppose that $f^p,h$ are Lipschitz continuous. 
Then Assumption~\ref{ass:obs} holds with the quadratic input-output stage cost~\eqref{eq:ell_reg}.
\end{proposition}
\begin{proof}
\begin{subequations}
\label{eq:obs_proof}
Consider $z^p_t=\pi_x(w_t)$, $v_t=\pi_u(w_t)$ (Ass.~\ref{ass:regulator}) and $\tilde{x}^p_{t+1}=f^p(\tilde{x}^p_t,v_t,w_t)$, $x^p_0=\tilde{x}^p_0$, $x^p_{t+1}=f^p(x^p_t,u_t,w_t)$, which implies
\begin{align}
\label{eq:obs_proof_0}
\sigma({x}_0)\stackrel{\eqref{eq:sigma}}{=}\|x ^p_0-z^p_0\|^2\stackrel{\eqref{eq:reg2}, \eqref{eq:observable}}{\leq }c_{obs}\sum_{j=0}^{\nu-1}\|h(\tilde{x}^p_j,v_j,w_j)\|^2.
\end{align}
Lipschitz continuity of $f^p$ implies for any $k=0,\dots,\nu$:
\begin{align}
\label{eq:obs_proof_1}
\|x^p_k-\tilde{x}^p_k\|^2\leq c_1 \sum_{j=0}^{\nu-1}\|u_j-v_j\|^2, 
\end{align}
with some constant $c_1>0$. 
Similarly, Lipschitz continuity implies for any $k=1,\dots,\nu$: 
\begin{align}
\label{eq:obs_proof_2}
c_2\sigma(x_k)=&c_2\|x^p_k-z^p_k\|^2\leq 2c_2\|\tilde{x}^p_k-z^p_k\|^2+2c_2\|x_p^k-\tilde{x}^p_k\|^2\nonumber\\
\stackrel{\eqref{eq:obs_proof_1}}{\leq}& 2c_2L_f^{2k}\|x^p_0-z^p_0\|^2+2c_2c_1\sum_{j=0}^{\nu-1}\|u_j-v_j\|^2\nonumber\\
\leq & \sigma(x_0)+\sum_{j=0}^{\nu-1}\|u_j-v_j\|^2,
\end{align}
with the Lipschitz constant $L_f>0$ and $c_2=\frac{1}{2\max\{1,L_f^{2\nu},c_1\}}>0$, where we used $\|a+b\|^2\leq 2(\|a\|^2+\|b\|^2)$ in the first inequality. 
Using further Lipschitz continuity of $h$ and  $\|a+b\|^2\leq 2(\|a\|^2+\|b\|^2)$, we have
\begin{align}
\label{eq:obs_proof_3}
&\|h(\tilde{x}^p_j,v_j,w_j)\|^2\\
\leq&2\|h(x^p_j,u_j,w_j)\|^2+2L_h^2(\|\tilde{x}^p_j-x^p_j\|^2+\|u_j-v_j\|^2), \nonumber
\end{align}
for $j=0,\dots,\nu-1$. 
Combining all these inequalities, we arrive at
\begin{align*}
&c_2\sigma(x_\nu)\stackrel{\eqref{eq:obs_proof_2}}{\leq} \sigma(x_0)+\sum_{j=0}^{\nu-1}\|u_j-v_j\|^2\\
\stackrel{\eqref{eq:obs_proof_0}}{\leq}& \sum_{j=0}^{\nu-1}c_{obs}\|h(\tilde{x}^p_j,v_j,w_j)\|^2+\|u_j-v_j\|^2\\
\stackrel{\eqref{eq:obs_proof_3}}{\leq}&\sum_{j=0}^{\nu-1}2 c_{obs} (\|h(x^p_j,u_j,w_j)\|^2+L_h^2(\|\tilde{x}_j^p-x_j^p\|^2+\|u_j-v_j\|^2))\\
&+\|u_j-v_j\|^2\\
\stackrel{\eqref{eq:obs_proof_1}}{\leq}&\sum_{j=0}^{\nu-1} 2c_{obs}\|h(x^p_j,u_j,w_j)\|^2+(1+2c_{obs}L_h^2(1+\nu c_1))\|u_j-v_j\|^2\\
\stackrel{\eqref{eq:ell_reg}}{\leq }&\max\left\{\frac{2c_{obs}}{\lambda_{\max}(Q)},\frac{1+2c_{obs}L_h^2(1+\nu c_1)}{\lambda_{\min}(R)}\right\} \sum_{j=0}^{\nu-1}\ell(x_j,u_j). &\qedhere
\end{align*}
\end{subequations}
\end{proof}
We note that Condition~\eqref{eq:obs} is analogous to final-state observability (cf~\cite[Def.~4.29]{rawlings2017model}) and hence the proof of Proposition~\ref{prop:obs_plant} is analogous to the more general result in~\cite[Prop.~4.31]{rawlings2017model}, which showed that given a (uniform) continuity bound, observability implies final-state observability. 
 \begin{remark}(Input-output models)
Consider for simplicity $\ell(x,u)=\|y\|^2_Q+\|u\|_R^2$ with some output $y=h(x,u)$. 
In the following, we show how the considered detectability and observability conditions trivially hold for final-state observable systems with the (typically non-minimal) state $x_t = (Q^{1/2}y_{t-\nu},\dots Q^{1/2}y_{t-1},R^{1/2}u_{t-\nu},\dots, R^{1/2}u_{t-1})$, $\sigma(x)=\|x\|^2$. 
Similar state-space representations are often considered for Nonlinear AutoregRessive models with eXogeneous inputs (NARX) and data-driven input-output models, compare, e.g.,~\cite{bongard2021robust}. 
First, note that Assumption~\ref{ass:obs} is trivially satisfied with $c_o=1$. 
Furthermore, a suitable storage function satisfying Assumption~\ref{ass:detect} can be constructed using
\begin{align}
W(x_t)=\dfrac{1}{\nu}\sum_{k=1}^\nu (\nu+1-k)\ell(x_{t-k},u_{t-k}),
\end{align}
which satisfies~\eqref{eq:detect_grimm} with $\gamma_o=1$, $\epsilon_o=1/\nu$.
\end{remark}

 \begin{remark}
\label{rk:flat}
(Flat systems and observability)
Consider the setup in Section~\ref{sec:minPhase}.
In the special case that $y$ is a flat output, we have no zero dynamics, i.e. $d=n_p-1$ in Assumption~\ref{ass:BINF}. 
Thus, the system reduces to an FIR filter with the input $F_1$. 
Hence, similar to Prop.~\ref{prop:obs_plant} and Prop.~\ref{prop:minphase_detect}, one can show that the stage cost $\ell_d$ satisfies the stronger observability condition (Ass.~\ref{ass:obs}) with $c_o=1$, $\nu=d+1=n_p$.
\end{remark} 

\clearpage
\subsection{Error feedback output regulation} 
\label{app:error_feedback}
In this section, we show how the results in Sections~\ref{sec:reg}--\ref{sec:input} can be extended to error feedback output regulation, i.e., when only a noisy output can be measured. 
The setup and observer are discussed in Section~\ref{sec:error_feedback_1}, continuity bounds are derived in Section~\ref{sec:error_feedback_2}, finite-gain $\mathcal{L}_2$ stability is shown in Section~\ref{sec:error_feedback_3} and some discussion on the results can be found in Section~\ref{sec:error_feedback_4}. 
%
\subsubsection{Error-feedback -  Setup and ISS observer}
\label{sec:error_feedback_1}
We consider the case where only a noisy output $\tilde{y}_t=y_t+\eta_t$ can be measured, 
with some noise $\eta_t$. 
While the problem of output feedback is always of relevance in MPC~\cite{findeisen2003state}, this is particularly the case in output regulation, which is classically solved with a dynamic error feedback~\cite{isidori1990output}. 
The basic idea is to use an observer to obtain estimates $(\hat{x}^p,\hat{w})$ using ($\tilde{y},u$) and implement a certainty equivalence MPC.
Then, by combining stability properties of the observer (Ass.~\ref{ass:exp_obs}), with continuity and  stability properties of the nominal output regulation MPC (cf. Thm.~\ref{thm:MPC} and Lemma~\ref{lemma:cont} below), we show finite-gain $\mathcal{L}_2$ stability in Theorem~\ref{thm:dyn_error}. 

We consider the following simplifying assumptions. 
\begin{assumption}
\label{ass:no_state_con}
(No state constraints)
The constraint set is given by  $\mathcal{Z}^p=\mathbb{R}^{n_p}\times\mathbb{U}$.
\end{assumption}
\begin{assumption}
\label{ass:increm_stable}
(Incrementally stable plant)
Assumption~\ref{ass:increm_stab} holds with $\kappa(x^p,t)=v_t$ and $\delta_{loc}=\infty$. 
\end{assumption}
%
\begin{assumption}
\label{ass:increm_w}
(Stable exosystem)
There exists a continuous incremental Lyapunov function $V_w:\mathbb{W}\times\mathbb{W}\rightarrow\mathbb{R}_{\geq 0}$ and constants $c_{l,w},c_{u,w}>0$, such that
\begin{subequations}
\label{eq:increm_w}
\begin{align}
\label{eq:increm_w_1}
c_{l,w}\|w_1-w_2\|\leq V_w(w_1,w_2)\leq & c_{u,w}\|w_1-w_2\|,\\
\label{eq:increm_w_2}
V_w(s(w_1),s(w_2))\leq& V_w(w_1,w_2).
\end{align}
\end{subequations}
\end{assumption}
The relaxation of Assumptions~\ref{ass:no_state_con}--\ref{ass:increm_stable}  requires additional tools from robust MPC to ensure recursive feasibility, constraint satisfaction and stability, which are beyond the scope of this paper and will be briefly discussed in Remark~\ref{rk:robust_output} below. 
Assumption~\ref{ass:increm_w}, similar to the classical conditions in~\cite[(A1)]{castillo1993nonlinear}, \cite[H1]{isidori1990output}, \cite[Ass.~A4.]{magni2001output}, is natural since an unstable exosystem cannot be stabilized and an asymptotically stable exosystem can be neglected, and hence via~\eqref{eq:increm_w} we only assume Lyapunov stability. 

For simplicity, we consider a  Luenberger-like observer with state estimate $\hat{x}_t=(\hat{x}_t^p,\hat{w}_t)\in\mathbb{R}^n$ and observer gain $L_t\in\mathbb{R}^{n\times p}$:
\begin{align}
\label{eq:obs_Luenberg}
\hat{x}_{t+1}={f}(\hat{x}_t,u_t)+L_t\cdot (\tilde{y}_t-h(\hat{x}^p_t,u_t,\hat{w}_t)),
\end{align} 
which includes standard observer designs like extended Kalman Filter (EKF) or high-gain observers.  
We denote the observer error by $\hat{e}_t=x_t-\hat{x}_t$. 
The following assumption captures the desired properties of the observer, similar to~\cite[Ass.~1]{kohler2019simple}.
\begin{assumption}
\label{ass:exp_obs} (ISS observer)
There exist constants $\gamma_1$, $\gamma_2$, $\gamma_3$, $\gamma_4>0$, $\rho\in(0,1)$, such that for any trajectory satisfying $(x^{p}_t,u_t)\in\mathcal{Z}^{p}$ and \eqref{eq:sys} for all $t\geq 0$, there exists an incremental Lyapunov function $\hat{V}_{o}:\mathbb{R}^n\times\mathbb{N}\rightarrow\mathbb{R}_{\geq 0}$ satisfying the following inequalities for all $t\geq 0$
\begin{subequations}
\label{eq:obs_exp}
\begin{align}
\label{eq:obs_exp_1}
\gamma_1\|\hat{e}_t\|^2\leq \hat{V}_o(\hat{x}_t,t)\leq& \gamma_2\|\hat{e}_t\|^2,\\
\label{eq:obs_exp_2}
\hat{V}_o(\hat{x}_{t+1},t+1)\leq& \rho\hat{V}_o(\hat{x}_t,t)+\gamma_3\|\eta_t\|^2,\\
\label{eq:obs_exp_3}
 \|\hat{x}_{t+1}-f(\hat{x}_t,u_t)\|^2\leq& \gamma_4 \hat{V}_o(\hat{x}_t,t)+\gamma_5\|\eta_t\|^2, 
\end{align}
\end{subequations}
with $\hat{x}_{t+1}$ according to~\eqref{eq:obs_Luenberg}. 
\end{assumption}
Conditions~\eqref{eq:obs_exp_1}--\eqref{eq:obs_exp_2} imply that the observer error is input-to-state stable (ISS). 
Condition~\eqref{eq:obs_exp_3} directly follows from the structure\footnote{%
Condition~\eqref{eq:obs_exp_3} also follows from~\eqref{eq:obs_exp_1}--\eqref{eq:obs_exp_2}, if $f$ is Lipschitz continuous:
$\|\hat{x}_{t+1}-f(\hat{x}_t,u_t)\|^2
\leq 2\|\hat{e}_{t+1}\|^2+2\|f(x_{t},u_t)-f(\hat{x}_t,u_t)\|^2
\leq 2\frac{L_f^2+\rho}{\gamma_1}\hat{V}_o(\hat{x}_t,t)+\frac{2\gamma_3}{\gamma_1}\|\eta_t\|^2.$
Thus, the following results equally apply to more general observers, e.g.,  constrained moving horizon estimation (MHE)~\cite{muller2017nonlinear}.}~\eqref{eq:obs_Luenberg} if the observer-gain $L_t$ is uniformly bounded and $h$ is Lipschitz continuous. 
In order to design such a stable observer, a detectability condition similar to Assumption~\ref{ass:IOSS} for the state $x=(x^p,w)$ is necessary (cf.~\cite[Prop.~23]{sontag1997output}).  
 Similar detectability conditions are also required for error-feedback  output regulation in the classical setup, compare e.g. \cite[A3/H3]{isidori1990output}. 
\subsubsection*{Dynamic error feedback MPC}
Given the stable observer (Ass.~\ref{ass:exp_obs}),  error-feedback output regulation is simply achieved by replacing the measured state $x_t=(x^p_t,w_t)$ by the estimate $\hat{x}_t=(\hat{x}^p_t,\hat{w}_t)$ in the MPC problem~\eqref{eq:MPC}:
\begin{subequations}
\label{eq:MPC_dyn}
\begin{align}
\label{eq:MPC_dyn_1}
V_N(\hat{x}_t):=\min_{u_{\cdot|t}\in\mathbb{U}^N}& \sum_{k=0}^{N-1} \ell(x_{k|t},u_{k|t})\\
\label{eq:MPC_dyn_2}
\text{s.t. }&{x}_{0|t}=\hat{x}_t,~{x}_{k+1|t}=f(x_{k|t},u_{k|t}),\\
&k=0,\dots,N-1,\nonumber
\end{align}
\end{subequations}
with the stage cost~\eqref{eq:ell_reg}. 
Thus, in each time step $t$, the observer updates the estimate $\hat{x}_t$ using $(u_{t-1},\tilde{y}_{t-1})$, the optimization problem~\eqref{eq:MPC_dyn} is solved and the first part of the resulting optimal input trajectory is applied to the system, i.e. $u_t=u^*_{0|t}$.

\subsubsection{Continuity properties}
\label{sec:error_feedback_2}
In Lemmas~\ref{lemma:increm_full}--\ref{lemma:cont}, we establish incremental stability properties of the full state $x$ and subsequently continuity properties of the value function $V_N$, using the following additional (global) Lipschitz continuity conditions. 
\begin{assumption}
\label{ass:continuity}(Lipschitz continuity)
The function $f^p$ is Lipschitz continuous w.r.t. $w$, the output map $h$  and the regulator maps $\pi_x,\pi_u$ from Assumption~\ref{ass:regulator} are Lipschitz continuous, and $\sqrt{V_s}$ from Assumption~\ref{ass:increm_stable} is Lipschitz continuous w.r.t. $x^p$. 
\end{assumption}
\begin{lemma}
\label{lemma:increm_full}
Let Assumptions~\ref{ass:no_state_con}--\ref{ass:increm_w} and \ref{ass:continuity} hold.
Then the system~\eqref{eq:sys} is incrementally stable, i.e., there exists a constant  $c_{sw}>0$, such that for any two initial conditions $x_0,z_0\in\mathbb{R}^n$ and any input trajectory $u_t=v_t\in\mathbb{U}$, the resulting state trajectories $x_t,z_t$ satisfy
\begin{align}
\label{eq:increm_full}
\|x_t-z_t\|^2 \leq c_{sw} \|x_0-z_0\|^2.
\end{align}
\end{lemma}
\begin{proof}
Assumption~\ref{ass:increm_w} ensures $\|w_k-\tilde{w}_k\|\leq \dfrac{c_{u,w}}{c_{l,w}}\|w_0-\tilde{w}_0\|$ for any two initial conditions $w_0,\tilde{w}_0$ of the exosystem. 
Consider $z_t=(z^p_t,w_t)$, $x_t=(x^p_t,\tilde{w}_t)$ and $u_t=v_t$. 
Lipschitz continuity of $f^p,\sqrt{V_s}$ (Ass.~\ref{ass:continuity}) and Inequality~\eqref{eq:increm_a} from Assumption~\ref{ass:increm_stab}/\ref{ass:increm_stable} imply 
the existence of a positive constant $L_s$, such that
\begin{align*}
&\sqrt{V_s(x^p_{t+1},t+1)}=\sqrt{V_s(f^p(x^p_t,v_t,\tilde{w}_t),t+1)}\\
\leq &L_s\|w_t-\tilde{w}_t\|+\sqrt{V_s(f^p(x^p_t,v_t,{w}_t),t+1)}\\
\stackrel{\eqref{eq:increm_a},\eqref{eq:increm_w}}{\leq}& \sqrt{\rho_s} \sqrt{V_s(x^p_t,t)}+L_s\dfrac{c_{u,w}}{c_{l,w}}\|w_0-\tilde{w}_0\|,
\end{align*}
which recursively applied ensures 
\begin{align*}
\sqrt{V_s(x^p_t,t)}
\leq &\max\left\{\sqrt{V_s(x^p_0,0)},\frac{L_{sw}c_{u,w}}{(1-\sqrt{\rho_s})c_{l,w}}\|w_0-\tilde{w}_0\|\right\}
\end{align*}
for all $t\geq 0$. 
Thus, condition~\eqref{eq:increm_full} holds with $c_{sw}:=\frac{c_{u,w}^2}{c_{l,w}^2}+\frac{1}{c_{s,l}}\max\{\sqrt{c_{s,u}},\frac{L_{sw}c_{u,w}}{(1-\sqrt{\rho_s})c_{l,w}}\}^2$.
\end{proof}
\begin{lemma}
\label{lemma:cont}
Let Assumption~\ref{ass:continuity} hold. 
There exist constants $c_\ell$, $c_\sigma>0$, such that for any $x,\tilde{x}\in\mathbb{R}^n$, $u\in\mathbb{U}$ and all $\epsilon>0$, the following inequalities hold:
\begin{subequations}
\label{eq:cont}
\begin{align}
\label{eq:ell_cont}
\ell(x,u)\leq &(1+\epsilon)\ell(\tilde{x},u)+\frac{1+\epsilon}{\epsilon}c_\ell\|x-\tilde{x}\|^2,\\
\label{eq:sigma_cont}
\sigma(x)\leq& (1+\epsilon)\sigma(\tilde{x})+\frac{1+\epsilon}{\epsilon}c_\sigma\|x-\tilde{x}\|^2.
\end{align}
Suppose further that Assumptions~\ref{ass:no_state_con}--\ref{ass:increm_w} hold. 
There exists a constant $c_V>0$, such that the following continuity property holds for any $x,\tilde{x}\in\mathbb{R}^n$, $N\in\mathbb{N}$ and all $\epsilon>0$:
\begin{align}
\label{eq:value_cont}
V_N(x)\leq (1+\epsilon)V_N(\tilde{x})+N\frac{1+\epsilon}{\epsilon}c_V\|x-\tilde{x}\|^2.
\end{align}
\end{subequations}
\end{lemma}
\begin{proof}
\textbf{Part I. }
Cauchy-Schwarz and Young's inequality imply that for any $a,b\in\mathbb{R}^n$, $\epsilon>0$, $Q=Q^\top\succ 0$, we have 
\begin{align}
\label{eq:cauchy_schwarz}
\|a+b\|_Q^2\leq (1+\epsilon) \|a\|_Q^2+\frac{1+\epsilon}{\epsilon}\|b\|_Q^2.
\end{align}
Using $\ell$  quadratic and $h,\pi_u$ Lipschitz continuous, we have
\begin{align*}
\ell(x,u)\stackrel{\eqref{eq:cauchy_schwarz}}{\leq}&(1+\epsilon)\ell(\tilde{x},u)+\frac{1+\epsilon}{\epsilon}\|h(x^p,u,w)-h(\tilde{x}^p,u,\tilde{w})\|_Q^2\\
&+\frac{1+\epsilon}{\epsilon}\|\pi_u(w)-\pi_u(\tilde{w})\|_R^2\\
\leq & (1+\epsilon)\ell(\tilde{x},u)+\frac{1+\epsilon}{\epsilon}c_\ell\|x-\tilde{x}\|^2.
\end{align*}
with $c_\ell:=L_h^2\lambda_{\max}(Q)+L_{\pi_u}^2\lambda_{\max}(R)>0$, where $L_h,L_{\pi_u}$ are the Lipschitz constants of $h$ and $\pi_u$, respectively.
Similarly, $\pi_x$ Lipschitz continuous with $L_{\pi_x}$ and $\sigma$ quadratic ensure that condition~\eqref{eq:sigma_cont} holds with $c_\sigma:=2\max\{1,L_{\pi_x}^2\}$.\\
\textbf{Part II. }Consider two initial conditions $x_0,\tilde{x}_0$, subject to the same input trajectory $u \in\mathbb{U}^N$ resulting in state trajectories $x_t,\tilde{x}_t$. 
Lemma~\ref{lemma:increm_full} ensures that $\|x_t-\tilde{x}_t\|^2\leq c_{sw}\|x_0-\tilde{x}_0\|^2$. 
Using Inequality~\eqref{eq:ell_cont}, this ensures 
\begin{align*}
J_N(x_\cdot,u_\cdot)\leq (1+\epsilon)J_N(\tilde{x}_\cdot,u_\cdot)+N\frac{1+\epsilon}{\epsilon}c_{\ell} c_{sw}\|x_0-\tilde{x}_0\|^2.
\end{align*}
Assumption~\ref{ass:no_state_con} ensures that Problem~\eqref{eq:MPC} is feasible for all $x\in\mathbb{R}^n$ , $u\in\mathbb{U}^N$ and thus~\eqref{eq:value_cont} holds with $c_V:=c_lc_{sw}$.
\end{proof} 

\subsubsection{Stability analysis} 
\label{sec:error_feedback_3}
The following theorem summarizes the stability properties of the proposed MPC using inherent robustness properties. 
\begin{theorem}
\label{thm:dyn_error}
Let Assumptions~\ref{ass:regulator}, \ref{ass:IOSS} and \ref{ass:no_state_con}-\ref{ass:continuity} hold. 
There exists a constant $N_1>0$, such that for any $N\in\mathbb{N}$ satisfying $N>N_1$, the closed loop is finite-gain $\mathcal{L}_2$ stable, i.e., there exist constants $\gamma_{\mathcal{L}_2,N},C_N>0$, such that for all $K\in\mathbb{N}$:
\begin{align*}
\sum_{t=0}^{K} \|\hat{e}_t\|^2+\sigma(x_t)\leq C_N(\sigma(x_0)+\|\hat{e}_0\|^2)+\sum_{t=0}^K \gamma_{\mathcal{L}_2,N}\|\eta_t\|^2.
\end{align*} 
\end{theorem}
\begin{proof}
First, we show a bound on the value function $V_N(\hat{x})$, then we provide a Lyapunov function to show ISS of $x$ w.r.t the noise $\eta$ and the estimation error $\hat{e}_t$. 
Finally, we provide a joint Lapunov function for the state $x$ and estimation error $\hat{e}$ w.r.t. the noise $\eta$, which implies $\mathcal{L}_2$-stability.\\
\begin{subequations}
\label{eq:proof_error_feedback}
\textbf{Part I. }Similar to Proposition~\ref{prop:stab}, Assumptions \ref{ass:regulator} and~\ref{ass:no_state_con}--\ref{ass:increm_stable} ensure that Assumption~\ref{ass:stab} holds with $\delta_s=\infty$ and  $x_t$ replaced by $\hat{x}_t$.
Using the nominal analysis in Theorem~\ref{thm:MPC} Part II, with $\gamma_{\overline{Y}}=\gamma_s+\gamma_o$ (since $\delta_s=\infty$), 
we have  
\begin{align}
\label{eq:proof_error_feedback_1}
V_N(x^*_{1|t})+\ell(\hat{x}_t,u_t)\stackrel{\eqref{eq:Value_dec}}{\leq} V_N(\hat{x}_t)+\underbrace{\dfrac{\gamma_s(\gamma_s+\gamma_o)}{\epsilon_o(N-1)}}_{=:(1-\alpha_N)\epsilon_o}\sigma(\hat{x}_t),
\end{align} 
with $\alpha_N>0$ for $N>N_1:=1+\gamma_s(\gamma_s+\gamma_o)/\epsilon_o^2$. 
From the properties of the observer (Ass.~\ref{ass:exp_obs}), we know that
\begin{align}
\label{eq:proof_error_feedback_0}
\overline{e}_t:=\|\hat{x}_{t+1}-x^*_{1|t}\|^2\stackrel{\eqref{eq:obs_exp_1},\eqref{eq:obs_exp_3}}{\leq}  \gamma_4\gamma_2 \|\hat{e}_t\|^2+\gamma_5\|\eta_t\|^2.
\end{align}
The perturbed closed loop satisfies the following inequality similar to~\eqref{eq:proof_error_feedback_1}:
\begin{align}
\label{eq:proof_error_feedback_2}
&V_N(\hat{x}_{t+1})-V_N(\hat{x}_t)\nonumber\\
\stackrel{\eqref{eq:value_cont}}{\leq} &
(1+\epsilon_1)V_N(x^*_{1|t})-V_N(\hat{x}_t)\nonumber\\
&+N\dfrac{1+\epsilon_1}{\epsilon_1}c_V\|x^*_{1|t}-\hat{x}_{t+1}\|^2\nonumber\\
\stackrel{\eqref{eq:proof_error_feedback_1},\eqref{eq:proof_error_feedback_0}}\leq &\epsilon_1 V_N(\hat{x}_t)-(1+\epsilon_1)\ell(\hat{x}_t,u_t)+N\dfrac{1+\epsilon_1}{\epsilon_1}c_V\overline{e}_t\nonumber\\
&+(1+\epsilon_1)(1-\alpha_N)\epsilon_o\sigma(\hat{x}_t)\nonumber\\
\stackrel{\eqref{eq:stab_grimm}}{\leq}&-(1+\epsilon_1)\ell(\hat{x}_t,u_t)+N\dfrac{1+\epsilon_1}{\epsilon_1}c_V\overline{e}_t\nonumber\\
&+((1+\epsilon_1)(1-\alpha_N)\epsilon_o+\epsilon_1\gamma_s)\sigma(\hat{x}_t)\nonumber\\
\stackrel{\eqref{eq:ell_cont},\eqref{eq:sigma_cont}}{\leq}&-\ell(x_t,u_t)+\dfrac{1+\epsilon_1}{\epsilon_1}c_\ell\|\hat{e}_t\|^2+N\dfrac{1+\epsilon_1}{\epsilon_1}c_V\overline{e}_t\nonumber\\
&+(1+\epsilon_1)((1+\epsilon_1)(1-\alpha_N)\epsilon_o+\epsilon_1\gamma_s)\sigma(x_t)\nonumber\\
&+\dfrac{1+\epsilon_1}{\epsilon_1}((1+\epsilon_1)(1-\alpha_N)\epsilon_o+\epsilon_1\gamma_s)c_\sigma\|\hat{e}_t\|^2\nonumber\\
=:&-\ell(x_t,u_t)+c_1\|\hat{e}_t\|^2+c_2\overline{e}_t+c_3\sigma(x_t),
\end{align}
with constants $c_1,c_2,c_3>0$ in dependence of $\epsilon_1>0$ chosen later.  \\
\textbf{Part II. } Consider the following  ISS Lyapunov function $\hat{Y}_N(\hat{x}_t,x_t)=W(x_t)+V_N(\hat{x}_t)$. 
The closed loop satisfies
\begin{align}
\label{eq:proof_error_feedback_3}
&\hat{Y}_N(\hat{x}_{t+1},x_{t+1})-\hat{Y}_N(\hat{x}_{t},x_{t})\nonumber\\
\stackrel{\eqref{eq:detect_grimm_2},\eqref{eq:proof_error_feedback_2}}{\leq} &
c_1\|\hat{e}\|_t^2+c_2\overline{e}_t-(\epsilon_o-c_3)\sigma(x_t)\nonumber\\
\stackrel{\eqref{eq:proof_error_feedback_0}}{\leq} &\tilde{c}_1\|\hat{e}_t\|^2+\tilde{c}_2\|\eta_t\|^2-\tilde{\epsilon}_o\sigma(x_t),
\end{align}
with $\tilde{c}_1:=c_1+c_2\gamma_2\gamma_4$, $\tilde{c}_2=c_2\gamma_5$, $\tilde{\epsilon}_o=\epsilon_o-c_3$.
Note that $\tilde{\epsilon}_o=\epsilon_o(1+\epsilon_1)^2\alpha_N-(1+\epsilon_1)\epsilon_1\gamma_s$. 
Thus, choosing $\epsilon_1:=\min\{1,\frac{\alpha_N \epsilon_o}{2(2\gamma_s+3\epsilon_o(1-\alpha_N))}\}\in(0,1]$ ($\epsilon_1\leq 1$ is only employed to simplify the expressions), we have $\tilde{\epsilon}_o\geq \epsilon_o\alpha_N/2>0$ and \eqref{eq:proof_error_feedback_3} can be used to show that $\hat{Y}_N$ is an ISS Lyapunov function w.r.t. the noise $\|\eta_t\|$ and the estimation error $\|\hat{e}_t\|$.\\
\textbf{Part III. }Consider the joint Lyapunov function $Y_{N,o}(\hat{x}_t,x_t,t):=\hat{Y}_N(\hat{x}_t,x_t)+\hat{c}_o\hat{V}_o(\hat{x}_t,t)$, with $\hat{c}_o:=2\frac{\tilde{c}_1}{(1-\rho)\gamma_1}>0$. 
The decrease condition follows using
 \begin{align}
\label{eq:proof_error_feedback_7}
&{Y}_{N,o}(\hat{x}_{t+1},x_{t+1},t+1)-Y_{N,o}(\hat{x}_t,x_t,t)\nonumber\\
\stackrel{\eqref{eq:obs_exp_2},\eqref{eq:proof_error_feedback_3}}{\leq} &-\tilde{\epsilon}_o\sigma(x_t)-\tilde{c}_1\|\hat{e}_t\|^2+(\tilde{c}_2+\hat{c}_o\gamma_3)\|\eta_t\|^2.
\end{align}
An upper bound follows from
\begin{align}
\label{eq:proof_error_feedback_4}
&Y_{N,o}(\hat{x}_t,x_t,t)\\
\stackrel{\eqref{eq:stab_grimm},\eqref{eq:detect_grimm_1},\eqref{eq:obs_exp_1}}{\leq} &\gamma_o\sigma({x}_t)+\gamma_s\sigma(\hat{x}_t)+\hat{c}_o\gamma_2\|\hat{e}_t\|^2 \nonumber\\ 
\stackrel{\eqref{eq:sigma_cont}}{\leq} &  (\gamma_o+2\gamma_s)\sigma(x_t)+(2c_\sigma\gamma_s+\hat{c}_o\gamma_2)\|\hat{e}_t\|^2,\nonumber
\end{align}
where we used \eqref{eq:sigma_cont} with $\epsilon=1$. 
The lower bound follows with
\begin{align}
\label{eq:proof_error_feedback_5}
&Y_{N,o}(\hat{x}_t,x_t,t)\stackrel{\eqref{eq:obs_exp_1}}{\geq} \ell(\hat{x}_t,u_t)+W(x_t)+\hat{c}_o\gamma_1\|\hat{e}_t\|^2\nonumber\\
\stackrel{\eqref{eq:ell_cont}}{\geq}&\frac{1}{2}(\ell(x_t,u_t)+W(x_t))+(\hat{c}_o\gamma_1-c_\ell)\|\hat{e}_t\|^2\nonumber\\
\stackrel{\eqref{eq:detect_grimm_2}}{\geq} &\epsilon_o/2\sigma(x_t)+(\hat{c}_o\gamma_1-c_\ell)\|\hat{e}_t\|^2,
\end{align}
\end{subequations}
where we used~\eqref{eq:ell_cont} with $\epsilon=1$ and $W\geq 0$.
The lower bound is positive definite w.r.t. $(\sigma(x)+\|\hat{e}_t\|^2)$, i.e. $(\hat{c}_o\gamma_1-c_\ell)>0$ since $\hat{c}_o=2\frac{\tilde{c}_1}{(1-\rho)\gamma_1}>\frac{c_1}{\gamma_1}>\frac{c_\ell}{\gamma_1}>0$. 
Inequalities~\eqref{eq:proof_error_feedback_7}--\eqref{eq:proof_error_feedback_5} directly imply that $Y_{N,o}$ is an ISS Lyapunov function for the state $(\sigma(x),\|\hat{e}\|)^2$ and input $\|\eta_t\|^2$. 
Furthermore, since all the bounds are linear/quadratic, the ISS gain is linear and thus the $\mathcal{L}_2$-gain holds, compare e.g.~\cite[Thm.~1]{sontag1998comments}.
\end{proof}

\subsubsection{Discussion}
\label{sec:error_feedback_4}
\begin{remark}
\label{rk:output_novel}
(Novelty output-feedback MPC)
The stability analysis in Theorem~\ref{thm:dyn_error} has two main differences compared to most existing inherent robustness analyses for output-feedback MPC:\\
1. We establish strong stability properties (finite-gain $\mathcal{L}_2$):  
This is due to the usage of the quadratic bounds from Lemma~\ref{lemma:cont}, similar to the finite-gain stability in~\cite{lorenzen2019robust} for adaptive MPC.
For comparison, most robust MPC and output-feedback analyses use simple Lipschitz bounds resulting in general practical stability/nonlinear ISS gains (cf.~\cite{messina2005discrete}, \cite[Thm.~6]{magni2004stabilization}), or even assume that no noise is present (cf.~\cite{magni2001output,findeisen2003output}). \\
2. The value function $V_N$ is not a Lyapunov function: 
The nominal stability analysis in Theorem~\ref{thm:MPC} involves the storage function $W$, which only provides suitable bounds~\eqref{eq:detect_grimm_2} for the nominal dynamics (unless additional stronger continuity conditions are imposed for $W$). 
Thus, the continuity properties from  Lemma~\ref{lemma:cont} do not necessarily imply similar continuity conditions on the Lyapunov function $Y_N(x)=W(x)+V_N(x)$ from Theorem~\ref{thm:MPC}, which is why $\hat{Y}_N(\hat{x}_t,x_t)=W(x_t)+V_N(\hat{x}_t)$ is considered in Theorem~\ref{thm:dyn_error}. 
Hence, extending the analysis to include additive disturbances on the plant and/or exosystem would require additional continuity conditions on the storage function $W$ (Ass.~\ref{ass:detect}). 
Thus, also the general analysis in~\cite{roset2008robustness} is not applicable, which assumes ISS w.r.t. additive disturbances. \\
To the best knowledge of the authors, Theorem~\ref{thm:dyn_error} is the first result to establish \textit{finite-gain} stability for noisy output-feedback MPC and in particular for the considered setup\footnote{%
As discussed in~\cite[Remark~2]{grimm2005model}, if $W,\ell,f$ are continuous, $N$ finite and $\mathbb{X}=\mathbb{R}^n$, the system has some (possibly small) inherent robustness property.}, where the value function is not a Lyapunov function.
\end{remark}

\begin{remark}
\label{rk:uniform_bounds}
(Uniform $\mathcal{L}_2$-gain in $N$)
We point out that the $\mathcal{L}_2$-gain derived in Theorem~\ref{thm:dyn_error} is not uniform in $N$, but instead deteriorates for $N\rightarrow\infty$ due to the usage of the bound~\eqref{eq:value_cont}. 
While this is also the case for most existing inherent robustness/ISS results in MPC, e.g., most evident in the bounds in~\cite{yu2014inherent,limon2009input},  in MPC without terminal constraints the stability properties ideally improve as $N$ increases, compare e.g.~\cite{kohler2018novel}.
The additional difficulty in the considered setup is that the exosystem is only (marginally) stable (Ass.~\ref{ass:increm_w}), cannot be stabilized,  but is subject to estimation errors. 
We conjecture that a uniform $\mathcal{L}_2$-gain (for $N$ arbitrary large) can be shown by using a modified input candidate solution in the stability proof. 
\end{remark}

\begin{remark}
\label{rk:robust_output}
(Robust constraint satisfaction)
In the presence of hard state constraints and/or unstable systems (Ass.~\ref{ass:no_state_con} and/or Ass.~\ref{ass:increm_stable} does not hold), the certainty equivalence output-feedback MPC~\eqref{eq:MPC_dyn} is in general not recursively feasible and/or may not necessarily ensure stability.  
Given a known set bound on the initial estimation error $\hat{e}_0\in\mathbb{E}$ and the measurement noise $\eta_t\in\mathbb{D}$, this problem can be circumvented by using constraint tightening methods from robust (output-feedback) MPC, compare e.g.~\cite{kohler2019simple,kohler2018novel,chisci2002feasibility,mayne2009robust,koehler2021Output}.  
If the bounds $\mathbb{E},\mathbb{D}$ are such that Assumption~\ref{ass:regulator} remains true with the tightened constraints (compare, e.g., \cite[Ass.~4]{kohler2018novel}), the stability result from Theorem~\ref{thm:dyn_error} remains true, albeit with more conservative bounds $N_1,\alpha_N,\gamma_{\mathcal{L}_2,N}$.
\end{remark}

\begin{remark}
\label{rk:state_feedback}
 (Partial state feedback)
An important special case of the error-feedback setup is that in addition to the noisy output $\tilde{y}_t$ also the plant state $x^p$ can be measured, which simplifies the observer design. 
In particular, in case the plant dynamics $f^p$ are independent of the exogenous signal $w$, the problem corresponds to (partially) unknown references.
These simplified scenarios also allow for simpler solutions using a robust design, compare~\cite{monasterios2018model} and \cite{di2016reference,falugi2015model}.
\end{remark}

\begin{remark}
\label{rk:offset}
(Applicability to formulations in Sections~\ref{sec:minPhase} and \ref{sec:input})
The presented error-feedback analysis is equally applicable to the formulations in Sections~\ref{sec:minPhase} and \ref{sec:input}.
In particular, since $u_t$ is known, the state $\xi_t=(u_{t-1},\dots,u_{t-T})$ is known and thus the observer is also ISS for the augmented state $\tilde{x}=(x,\xi)$ (Ass.~\ref{ass:exp_obs} holds). 
Furthermore, the periodicity condition (Ass.~\ref{ass:exo_periodic}) in combination with $s$ Lipschitz continuous implies Assumption~\ref{ass:increm_w} and the two conditions are in fact equivalent for linear systems (cf. Sec.~\ref{sec:linear}).
We point out that for $s(w)=w$ (constant  exogenous signal $w$), the proposed error-feedback output regulation MPC provides a natural solution to the \textit{offset-free} control problem often considered in MPC~\cite{morari2012nonlinear}. 
\end{remark}

\begin{remark}
\label{rk:rexp_obs}
(Exponential-linear bounds)
In Sections~\ref{sec:reg}--\ref{sec:input} we used simple linear/quadratic/exponential bounds (Ass.~\ref{ass:increm_stab}, \ref{ass:IOSS}, \ref{ass:nonres_nonlin}), but the results for exact state measurements can be directly extended to ensure asymptotic stability with more general bounds using comparison functions, as, e.g., done in~\cite{grimm2005model}. 
However, for the nonlinear error/output-feedback problem considered in this section, the restriction to linear/quadratic bounds is \textit{not} without loss of generality, since the separation principle does not hold for general nonlinear systems, but, e.g., for nonlinear systems with linear/exponential bounds (cf., e.g.,~\cite{manchester2014output}).
\end{remark}

\begin{remark}
\label{rk:error_linear} (Linear case)
In the special case of linear systems (Sec.~\ref{sec:linear}), Assumption~\ref{ass:increm_stable} is equivalent to $A$ Schur and Assumption~\ref{ass:continuity} holds since $f^p,h,\pi_x,\pi_u$ are linear and $V_s$ is quadratic.
A  Luenberger observer  ($L_t$ constant in~\eqref{eq:obs_Luenberg}) satisfying Assumption~\ref{ass:exp_obs} can be designed, if $\begin{pmatrix}\begin{pmatrix}A&P_x\\0&S\end{pmatrix}&\begin{pmatrix}C&-P_y\end{pmatrix}\end{pmatrix}$ is detectable. 
For the special case of linear systems with constant exogenous signals ($S=I$), detectability is also treated extensively in~\cite{muske2002disturbance} for offset free tracking MPC. 
We point out that in the linear case, $W$ (Ass.~\ref{ass:detect}) is quadratic and hence the stability proof in Theorem~\ref{thm:dyn_error} can be extended to encompass additional random   disturbances.

Furthermore, in the linear unconstrained case the error feedback MPC reduces to a linear stabilizing feedback with a stable observer. 
Since $u_t=K_N\hat{x}_t$ is linear in the estimation error $\hat{e}_t$, which is ISS (Ass.~\ref{ass:exp_obs}), we get finite-gain $\mathcal{L}_2$ stability with a gain $\gamma_{\mathcal{L}_2,N}$ that improves as $N\rightarrow\infty$, unlike in Theorem~\ref{thm:dyn_error}.
\end{remark}

\end{document}